\documentclass[12pt]{article}

\usepackage[width=16cm, height=22cm]{geometry}   
\usepackage[font={small}, width=14cm]{caption}
\usepackage{amsmath}
\usepackage{amssymb}
\usepackage{amsthm}
\usepackage{amscd}
\usepackage{bbm}
\usepackage{enumerate}
\usepackage{float}
\usepackage{authblk}
\usepackage{xcolor}
\usepackage{mathrsfs}
\usepackage[mathscr]{eucal}
\usepackage[normalem]{ulem}
\usepackage{graphicx}
\usepackage{tikz-cd}
\usepackage{comment}
\usepackage{pgfplots}
\usepackage{subcaption}
\usepackage{empheq}
\usepackage{fancyvrb}
\usepackage{url}
\usepackage{todonotes}

\usepackage{verbatim}

\usetikzlibrary{cd}
\usetikzlibrary{patterns}

\newcounter{comments}
\newenvironment{displaycomment}{\begin{list}{}{\rightmargin=1cm\leftmargin=1cm}\item\sf\begin{small}}{\end{small}\end{list}}

\newtheorem{thm}{Theorem}[section]
\newtheorem{prop}[thm]{Proposition}
\newtheorem{lem}[thm]{Lemma}
\newtheorem{cor}[thm]{Corollary}

\theoremstyle{definition}
\newtheorem{dfn}[thm]{Definition}

\theoremstyle{remark}
\newtheorem{rem}[thm]{Remark}

\newtheorem{example}[thm]{Example}

\DeclareMathOperator{\diam}{\mathrm{diam}}
\newcommand{\Tr}{\mathrm{Tr}}
\newcommand{\CC}{\mathbb{C}}

\newcommand{\RR}{\mathbb{R}}

\newcommand{\ZZ}{\mathbb{Z}}

\newcommand{\Supp}{\mathrm{Supp}}

\newcommand{\sB}{\mathscr{B}}
\newcommand{\vol}{\mathrm{vol}}

\newcommand{\threebar}{\vert\kern-0.25ex\vert\kern-0.25ex\vert}

\title{Quantization of conductance and the coarse cohomology of partitions}

\author[1]{Matthias Ludewig}
\author[2]{Guo Chuan Thiang}
\affil[1]{University of Regensburg, Germany}
\affil[2]{Beijing International Center for Mathematical Research, Peking University, China}

\date{\today}

\begin{document}
\maketitle
\begin{abstract}
We demonstrate how integer quantization of Hall conductance arises from the large-scale geometry of the sample. Specifically, the Hall conductance is a higher-trace pairing of the Fermi projection with a coarse cohomology class coming from a partition of the geometric sample, whose integrality is proved.
\end{abstract}

\tableofcontents

\section{Introduction}
We introduce coarse cohomology techniques to the study of ``topological phases'' in physics, with particular emphasis on quantization of Hall conductance. 
It is built on two independently developed sets of ideas. The physical one concerns a position-space notion of Chern number for gapped 2D systems, proposed in Appendix C of \cite{Kitaev}. The mathematical one is coarse geometry, or ``large-scale geometry'', as developed by J.\ Roe, \cite{Roe-partition, Roe-cohom}, with a view to index theory on non-compact manifolds.

The main physical application is to the 2D quantum Hall effect (QHE), which is the paradigmatic and experimentally established example of a ``topological phase''. The adjective ``topological'' is a bit of a misnomer, because the phenomenon \emph{ignores} the topology of the sample surface $M$ (holes, lattice versus manifold, etc.), at least on the small scale. What is meant, in practice, is the robustness of the integer quantization of the Hall conductance, even as various laboratory parameters, including the surface geometry, are varied, often quite violently. A standard theoretical explanation of the quantization proceeds by identifying, via a Kubo formula, the Hall conductance with the Chern number of a certain vector bundle of occupied electron states fibred over the momentum space 2-torus \cite{TKNN,Simon}. However, this requires translation invariant setups in Euclidean plane geometry, whereas the QHE does not require such geometric idealizations. Indeed, the precision of the quantization (better than $\sim 10^{-9}$ \cite{JJ}) far exceeds the flatness of any 2D interfaces that can be fabricated \cite{Cheng}. For discretized models, one can even realize \emph{amorphous} topological insulators \cite{Mitchell, Agar} and topological metamaterials on aperiodic patterns \cite{AQPP}. It would seem, then, that no symmetry whatsoever is required to ``protect'' such topological phases --- the \emph{less} symmetry required, the \emph{more} robust and universal.

What, if not symmetry or topology, should underlie the quantization phenomenon? A germ of an answer was given in \cite{Kitaev}, in his ``Chern number for quasidiagonal projections'' over the lattice $M=\ZZ^2$. Partition $M$ into three pieces, $A,B,C$. Consider all triangles $\Delta$, with vertices $x_1,x_2,x_3\in M$, such that exactly one vertex lies in each of $A,B,C$, see Fig.~\ref{fig:loop}. Given a projection  (or merely idempotent) operator $P$ acting on $\ell^2(M)$, take the \emph{oriented} sum of the product of its hopping elements around such triangles,
\begin{equation}
\left(\sum_{{\rm anticlockwise}\;\Delta}-\sum_{{\rm clockwise}\;\Delta}\right)P(x_1,x_2)P(x_2,x_3)P(x_3,x_1),\label{eqn:Kitaev.sum}
\end{equation}
where $P(x,y)$ denotes the matrix element of $P$ from $y$ to $x$.

Closely related to Eq.\ \eqref{eqn:Kitaev.sum} is the following commutator-trace,
\begin{equation}
\Tr[P\chi_XP,P\chi_YP],\label{eqn:Kitaev.commutator.trace}
\end{equation}
where $X,Y\subseteq M$ are transverse half-spaces, and $\chi_X, \chi_Y$ are the respective multiplication operators by their characteristic functions. The Hall conductance is given by an expression such as Eq.\ \eqref{eqn:Kitaev.sum} or \eqref{eqn:Kitaev.commutator.trace}, see Section  \ref{sec:Hall.conductance}. 
Our main Theorem \ref{thm:quantization.RD} establishes the following surprising properties of this conductance:
\begin{itemize}
\item It is finite, for a large class of infinite-rank $P$. [Measurable quantity]
\item It can be non-zero. [Chiral asymmetry]
\item It is an integer multiple of an $M$-independent constant. [Universal quantization]
\item It is stable against variations of the partition choice. [Cobordism invariance]
\item It is stable against continuous variations of $P$. [Stability of quantization]
\end{itemize}
Specifically, we will explain how Eq.\ \eqref{eqn:Kitaev.sum}--\eqref{eqn:Kitaev.commutator.trace} are partition-idempotent pairings in a coarse (co)homology sense. Thus the pairing depends only on the \emph{large-scale} geometry of $X,Y\subseteq M$, and large-scale aspects of $P$. 
Although large-scale geometry is, informally, about geometry ``at infinity'', our work has implications for approximate quantization results on large but bounded $M$ as well, as discussed in Section \ref{sec:finite.size}.
The filtering-out of small-scale details should be compared with effective theory approaches, where it is implicit that large-scale contributions are dominant, and constrained by gauge-invariance principles to be governed by topological field theories \cite{Frolich}.

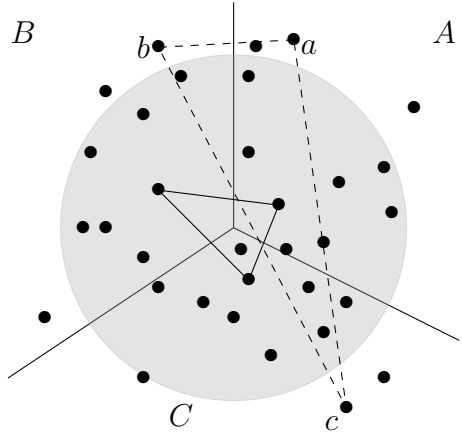
\begin{figure}
\centering
\begin{tikzpicture}
\draw (-3,-2)--(0,0);
\draw (0,0)--(3,-1.5);
\draw (0,0)--(0,3);
\node at (2.8,2.6) {$A$};
\node at (-2.8,2.6) {$B$};
\node at (-0.7,-2.5) {$C$};
\filldraw[color=gray!70,opacity=0.3] (0,0) circle (2.3cm);
\node at (1,2.4) {$a$};
\node at (-1.2,2.4) {$b$};
\node at (1.3,-2.6) {$c$};
\draw (0.6,0.3) -- (-1,0.5) -- (0.2,-0.7) -- cycle;
\draw[dashed] (0.8,2.5) -- (-1,2.4) -- (1.5,-2.4) -- cycle;
\node at (0.6,0.3) {$\bullet$};
\node at (-1,0.5) {$\bullet$};
\node at (0.2,-0.7) {$\bullet$};
\node at (0.8,2.5) {$\bullet$};
\node at (-1,2.4) {$\bullet$};
\node at (1.5,-2.4) {$\bullet$};
\node at (0.2,1) {$\bullet$};
\node at (0.2,2) {$\bullet$};
\node at (0.3,2.4) {$\bullet$};
\node at (1.4,0.6) {$\bullet$};
\node at (2,0.8) {$\bullet$};
\node at (2.1,0.2) {$\bullet$};
\node at (2.4,1.6) {$\bullet$};
\node at (1.2,-0.2) {$\bullet$};
\node at (1.5,-1) {$\bullet$};
\node at (2,-2) {$\bullet$};
\node at (1.2,-1.4) {$\bullet$};
\node at (1,-0.8) {$\bullet$};
\node at (0.7,-0.3) {$\bullet$};
\node at (0.5,-1.7) {$\bullet$};
\node at (0.1,-0.3) {$\bullet$};
\node at (-0.4,-1) {$\bullet$};
\node at (0,-1.2) {$\bullet$};
\node at (-1,-0.8) {$\bullet$};
\node at (-1.2,-0.4) {$\bullet$};
\node at (-1.2,-2) {$\bullet$};
\node at (-2.51,-1.2) {$\bullet$};
\node at (-1.7,0) {$\bullet$};
\node at (-2,0) {$\bullet$};
\node at (-1.9,1) {$\bullet$};
\node at (-1.7,1.8) {$\bullet$};
\node at (-1.2,1.5) {$\bullet$};
\node at (-0.7,2) {$\bullet$};
\end{tikzpicture}
\caption{Contributions to the pairing $\langle A,B,C;P\rangle$ come from triangular loops, with one vertex lying in each component of the partition of $M$. Triangles with short edges lie within the shaded area near the triple intersection of the partition. Triangles with long edges have highly suppressed contributions, because the hopping terms of $P$ decrease rapidly with respect to hopping distance.}\label{fig:loop}
\end{figure}

\paragraph{Historical remarks.}
A cyclic cohomology--$K$-theory pairing was used in the influential papers of Bellissard et al \cite{BES}, which introduced Connes' noncommutative differential geometry techniques to the study of the QHE. There, the focus was rather on proving integrality for a fixed Euclidean sample geometry. A hyperbolic plane geometry version was similarly considered in \cite{CHMM}. The work of Connes was partly influenced by the classic Helton--Howe work on traces of commutators \cite{HH}, and commutator-traces are closely related to Hall conductance. The development of coarse (co)homology and index theory followed soon after \cite{Roe-cohom, Roe-partition}. 

We mention \cite{HL} which studies a certain Bott index for non translation-invariant lattice models on compact surfaces. 
For lattice models on $\ZZ^2$, standard half-spaces, and finite-range Hamiltonians, an argument in the spirit of Section \ref{sec:PI.pairing} was sketched in \S6 of \cite{Ton}, for the study of the QHE Chern numbers.
The papers \cite{KS, KS2} apply coarse cohomology ideas, but in a very different way, for the study of other types of topological phases.

\section{Partition-idempotent pairing}\label{sec:PI.pairing}

\subsection{An abstract quantization argument}
Let $\mathscr{H}$ be a (complex) Hilbert space.
If $P$ and $X$ are bounded operators on $\mathscr{H}$, we will write
\[
P_X:=PXP.
\]

\begin{prop}\label{prop:minimal.quantization}
Suppose $P$, $X$ and $Y$ are bounded operators on $\mathscr{H}$ such that
\begin{equation}
\label{eqn:minimal.tc.assumptions}
[P_X,P_Y],\qquad \text{and} 
\qquad 
\begin{cases}
(P_X-P_X^2)(P_Y-P_Y^2)\\
(P_Y-P_Y^2)(P_X-P_X^2)\\
(P_X-P_X^2)^*(P_Y-P_Y^2)\\
(P_Y-P_Y^2)(P_X-P_X^2)^*
\end{cases}
\qquad 
\text{are trace class.}
\end{equation}
 Then
\[
2\pi i\cdot\Tr[P_X,P_Y]\;\in\;\ZZ.
\]
\end{prop}

\begin{proof} 
Consider the holomorphic function $\phi:z\mapsto e^{2\pi iz}-1$ on $\CC$, which has zeroes at integer values of $z$. We can factorize $\phi$ as
\[
\phi(z)=\psi(z)\cdot z(1-z)=\psi(z)\cdot(z-z^2)=(z-z^2)\cdot\psi(z)
\]
for some other holomorphic $\psi$. Functional calculus of $P_X$ and $P_Y$ gives
\begin{align*}
e^{2\pi iP_X}-1&=\phi(P_X)=\psi(P_X)(P_X-P_X^2)=(P_X-P_X^2)\psi(P_X),\\
e^{2\pi iP_Y}-1&=\phi(P_Y)=\psi(P_Y)(P_Y-P_Y^2)=(P_Y-P_Y^2)\psi(P_Y),\\
\end{align*}
Then, due to assumption \eqref{eqn:minimal.tc.assumptions}, the operators
\begin{equation}
\label{eqn:product.exponential.trace.class}
\begin{aligned}
(e^{2\pi i P_X}-1)(e^{2\pi i P_Y}-1)&=\psi(P_X)\cdot (P_X-P_X^2)(P_Y-P_Y^2)\cdot \,\psi(P_Y),\\
(e^{2\pi i P_Y}-1)(e^{2\pi i P_X}-1)&=\psi(P_Y)\cdot (P_Y-P_Y^2)(P_X-P_X^2)\cdot \,\psi(P_X),\\
(e^{2\pi i P_X}-1)^*(e^{2\pi i P_Y}-1)&=\psi(P_X)^*\cdot (P_X-P_X^2)^*(P_Y-P_Y^2) \cdot \,\psi(P_Y),\\
(e^{2\pi i P_Y}-1)(e^{2\pi i P_X}-1)^*&=\psi(P_Y)\cdot (P_Y-P_Y^2)(P_X-P_X^2)^* \cdot \,\psi(P_X)^*,
\end{aligned}
\end{equation}
are trace class.

By a remarkable observation of Kitaev (recalled in Prop.\ \ref{prop:Kitaev.conjecture} below), the trace class properties in Eq.\ \eqref{eqn:product.exponential.trace.class} imply the triviality of the following Fredholm determinant\footnote{This is the well-known generalization of the determinant to operators of the form 1+(trace class).},
\begin{equation}
\det\big(e^{2\pi iP_X}e^{2\pi iP_Y}e^{-2\pi iP_X}e^{-2\pi iP_Y}\big)=1.\label{eqn:vanishing.Fredholm.det}
\end{equation}
Then, by a well-known identity of Pincus--Helton--Howe (Prop.\ 10.1 of \cite{HH}, Eq.\ (1.3) of \cite{EF}), Eq.\ \eqref{eqn:vanishing.Fredholm.det} can be rewritten as
\[
1=\det\big(e^{2\pi iP_X} e^{2\pi iP_Y} e^{-2\pi iP_X} e^{-2\pi iP_Y}\big)=\exp\bigl(\Tr[2\pi iP_X,2\pi iP_Y]\bigr).
\]
Therefore $(2\pi i)^2\cdot \Tr[P_X,P_Y]\in 2\pi i\cdot\ZZ$, as claimed.
\end{proof}

\begin{prop}[Kitaev observation]\label{prop:Kitaev.conjecture}
Let $S,T$ be invertible operators on a Hilbert space. Suppose the following products are trace class,
\[
(S-1)(T-1),\qquad (T-1)(S-1)
\qquad  (S^*-1)(T-1),\qquad (T-1)(S^*-1).
\]
Then the following Fredholm determinant is well-defined and trivial,
\[
\det(STS^{-1}T^{-1})=1.
\]
\end{prop}
\begin{proof}
A version of this claim appeared as Eq.\ (133) of \cite{Kitaev}, with a rigorous proof recently supplied in \cite{EF}, Theorem 1.2.
\end{proof}

\begin{rem}
If $P,X,Y$ are self-adjoint, it suffices that $[P_X,P_Y]$ and $(P_X-P_X^2)(P_Y-P_Y^2)$ are trace class, to deduce $2\pi i\cdot \Tr[P_X,P_Y]\in\ZZ$.
\end{rem}

\begin{rem}
One reason to isolate the abstract Proposition \ref{prop:minimal.quantization} is to clarify that the quantization of possible values of $\Tr[P_X,P_Y]$ does \emph{not} require ``topology'' in the sense that a vector bundle Chern number approach (see Introduction) might suggest.
Instead, we will construct operator algebras whose idempotents $P$, together with certain geometrically determined projections $X,Y$, satisfy the conditions of Proposition \ref{prop:minimal.quantization}. 
Where appropriate, the operator algebra has a topology with respect to which $P\mapsto \Tr[P_X,P_Y]$ is continuous, thus locally constant. 
In applications, one would like such an operator algebra to be affiliated to the physical Hamiltonian operators at hand, in the sense of containing the spectral projections for compact, separated parts of the spectrum.
This is satisfied in relevant examples; see Section \ref{sec:Hall.conductance}.
\end{rem}

\subsection{Coarsely transverse half-spaces and partitions}
Subsequently, we assume that $(M,\mu, d)$ is a metric measure space, which is \emph{proper} in the sense that closed and bounded subsets are compact. 
Given a subset $Z\subseteq M$, the set of points lying within distance $r$ from $Z$ is denoted
\[
Z_r:=\{x\in M\,:\,d(x,Z)\leq r\}.
\]
We note that if $A, B \subseteq M$ are subsets, we have

\begin{equation*}
(A \cup B)_r = A_r\cup B_r,\qquad \text{and} \qquad (A \cap B)_r \subseteq A_r\cap B_r,
\end{equation*}
which we will use without comment throughout.

\begin{dfn}\label{dfn:transverse.subsets}
A collection of subsets $A_0,\ldots,A_q\subseteq M$ is \emph{coarsely transverse} if
\begin{equation*}
(A_0)_r\cap\ldots \cap (A_q)_r\;\;\;\mathrm{is\;bounded}\;\;\forall r>0.
\end{equation*}
\end{dfn}

\begin{dfn}\label{dfn:transverse.half.spaces}
We call a pair of subsets $X,Y\subseteq M$ \emph{coarsely transverse half-spaces} if the subsets $X,X^c,Y,Y^c$ are coarsely transverse in the sense of Definition \ref{dfn:transverse.subsets}.
\end{dfn}

\begin{example}
For $M=\RR^2$ the Euclidean plane, the Cartesian coordinate half-spaces $X=\{(x,y):x\geq 0\}$ and $Y=\{(x,y):y\geq 0\}$ are coarsely transverse. 
In \cite{SZ}, a related notion of \emph{multi-partitioned manifolds} is studied. 
\end{example}

Closely related to the notion of coarsely transverse half-spaces is that of a partition:
\begin{dfn}\label{dfn:partition}
For $q\geq 0$, a \emph{$q$-partition} of $M$ is an ordered, coarsely transverse collection $(A_0,\ldots,A_q)$ of disjoint Borel subsets whose union is $M$.
\end{dfn}

\begin{example}
Only compact spaces $M$ admit $0$-partitions, namely $(M)$ itself. The \emph{partitioned manifolds} studied by J.\ Roe in \cite{Roe-partition} provide examples of 1-partitions. This is a Riemannian manifold $M=M^+\cup M^-$ with $N=M^+\cap M^-=\partial M^+=\partial M^-$ a compact hypersurface. Then $(M^+,M^-\setminus N)$ is a 1-partition of $M$. Some 2-partitions are illustrated in Fig.\ \ref{fig:subdivision}. 
\end{example}

\begin{example}\label{ex:Euclidean.simplex}
For $M=\RR^q$, let $x_0, \dots, x_q \in \RR^q$ be $q+1$ points such that zero lies in the interior of their convex hull. Then
\begin{equation*}
A_i = \left\{\sum_{k=0}^q \lambda_k x_k \mid \lambda_k \in \RR_{\geq 0}, \lambda_i = 0\right\},\qquad i=0,\ldots, q,
\end{equation*}
gives a $q$-partition of $\RR^q$. (Points lying on an overlap may be assigned to any of the $A_i$ that they belong to.)
\end{example}

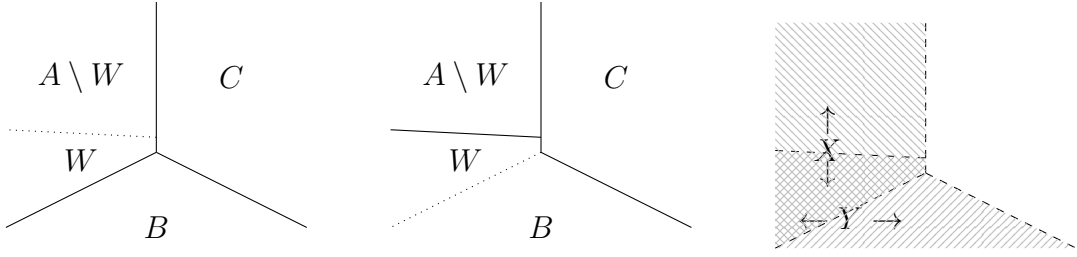
\begin{figure}
\centering
\begin{tikzpicture}
\draw (-2,-1)--(0,0);
\draw (0,0)--(2,-1);
\draw (0,0)--(0,2);
\draw[dotted] (0,0.2)--(-2,0.3);
\node at (-1,-0.1) {$W$};
\node at (-1,1) {$A\setminus W$};
\node at (0,-1) {$B$};
\node at (1,1) {$C$};
\end{tikzpicture}
\hspace{2em}
\begin{tikzpicture}
\draw (2,-1)--(0,0);
\draw[dotted] (0,0)--(-2,-1);
\draw (0,0)--(0,2);
\draw (0,0.2)--(-2,0.3);
\node at (0,-1) {$B$};
\node at (-1,1) {$A\setminus W$};
\node at (-1,-0.1) {$W$};
\node at (1,1) {$C$};
\end{tikzpicture}
\hspace{2em}
\begin{tikzpicture}
\draw[dashed] (-2,-1)--(0,0);
\draw[dashed] (0,0)--(2,-1);
\draw[dashed] (0,0.2)--(0,2);
\draw[dashed] (0,0)--(0,0.2);
\draw[dashed] (0,0.2)--(-2,0.3);
\node at (-1.3,0.3) {$X$};
\node at (-1.3,0.7) {$\uparrow$};
\node at (-1.3,0) {$\downarrow$};
\node at (-1,-0.6) {$\leftarrow Y\rightarrow$};
\fill[pattern=north west lines, pattern color=gray!50] (0,0) -- (0,2) -- (-2,2) -- (-2,-1) -- (0,0);
\fill[pattern=north east lines, pattern color=gray!50] (0,0) -- (2,-1) -- (-2,-1) -- (-2,0.3) -- (0,0.2) -- (0,0);
\end{tikzpicture}
\caption{The 2-partitions $(A,B,C)$ and $(A\setminus W,B\sqcup W,C)$ of the plane are cobordant. They are associated to the pair of coarsely transverse half-spaces $X,Y$.}\label{fig:subdivision}
\end{figure}

The following two Lemmas relate the notions of 2-partitions and coarsely transverse half-spaces.

\begin{lem}\label{lem:partition.from.multipartition}
Let $X,Y\subseteq M$ be coarsely transverse half-spaces. Then
\begin{equation*}
(X, X^c\cap Y, X^c\cap Y^c).
\end{equation*}
is a 2-partition of $M$.
\end{lem}
\begin{proof}
We have
\begin{equation*}
X_r\cap (X^c\cap Y)_r\cap (X^c\cap Y^c)_r\;\subseteq X_r\cap (X^c)_r\cap Y_r\cap (Y^c)_r,
\end{equation*}
with the latter bounded by definition of $X,Y$ being coarsely transverse.
\end{proof}

We can also go from a 2-partition to a pair of coarsely-transverse half-spaces:

\begin{lem}\label{lem:partition.to.multiplatition.rd}
Let $(A,B,C)$ be a 2-partition of $M$. Define
\[
W=\{a\in A\,:\,d(a,B)\leq d(a,C)\}\subseteq A.
\]
Then
\[
X=A,\qquad Y=W\sqcup B
\]
are coarsely-transverse half-spaces.
\end{lem}

\begin{rem}
\label{RemarkBackAndForth}
Observe that if we start with a 2-partition $(A, B, C)$, form the corresponding pair of coarsely transverse half-spaces $X$, $Y$ as in Lemma~\ref{lem:partition.to.multiplatition.rd} and then take the 2-partition determined by these half-spaces as in Lemma~\ref{lem:partition.from.multipartition}, we get back $(A, B, C)$.
The converse is generally false.
\end{rem}

\begin{proof}
Let $r>0$ and suppose $x\in W_r\cap C_r$. There exists $a \in W$ such that $d(x, a) \leq 2r$.
Then
\begin{align*}
d(x,B)&\leq d(x,a)+d(a,B) \leq 2r+d(a,C) \leq 2r+d(a,x) +d(x,C)\leq 6r,
\end{align*}
thus $W_r\cap C_r \subseteq B_{6r}$.
On the other hand, the same intersection is trivially contained in $A_{6r}$ and $C_{6r}$, hence it is bounded since $(A,B,C)$ forms a 2-partition.
Similarly, 
\[
(A\setminus W)_r\cap B_r
\]
is bounded. 
We now have that
\begin{equation*}
X_r\cap(X^c)_r\cap Y_r\cap(Y^c)_r\;\subseteq\; A_r\cap(B\sqcup C)_r\cap (W\sqcup B)_r\cap ((A\setminus W)\sqcup C)_r,
\end{equation*}
which by the observations before is a finite union of bounded subsets, thus bounded.
\end{proof}

\subsection{Partition-idempotent pairing}\label{sec:fp.pairing}

With $\mu$ a Borel measure on $M$, the Hilbert space $L^2(M)\equiv L^2(M,\mu)$ comes with a representation of the algebra $B(M)$ of bounded Borel functions on $M$.
Throughout, abusing notation, we write $A$ for the operator that multiplies by the characteristic function of a Borel subset $A \subseteq M$.

\begin{dfn} We write $\sB_{\rm fin}(M)$ for the $*$-subalgebra of bounded operators on $L^2(M)$ comprising operators $L$ which are:
\begin{enumerate}[(i)]
\item 
\emph{Locally trace class}: $KL$ and $LK$ are trace class for all bounded Borel subsets $K \subseteq M$.
\item  
\emph{Finite-propagation}: there exists $r>0$ such that whenever two bounded Borel subsets $A, B \subseteq M$  have $d(A,B)>r$, then $ALB=0$. 
\end{enumerate}
\end{dfn}

Finite propagation, or ``finite hopping range'', means that the integral kernel of $L$ is supported within some finite distance $r$ from the diagonal of $M\times M$.

For Borel subsets $A, B, C \subseteq M$ and an idempotent $P = P^2 \in \mathcal{L}(L^2(M))$, define a new bounded operator on $L^2(M)$ by the formula\footnote{Recall that we write $A$ for the operator that multiplies by the characteristic function of $A \subseteq M$.}
\begin{equation}
\label{generalizedCommutator}
\begin{aligned}
[ A, B, C]_P &:=  APBPCP+BPCPAP+CPAPBP \\
&\qquad-CPBPAP-BPAPCP-APCPBP.
\end{aligned}
\end{equation}
By definition, the expression $[ A, B, C]_P$ is anti-symmetric in $A$, $B$ and $C$. 
In particular, $[ A, B, C]_P = 0$ whenever two of the operators $A$, $B$ and $C$ are equal.
Moreover, if $A$ and $A^\prime$ are disjoint subsets of $M$, we have
\begin{equation}
\label{AdditivityGeneralizedCommutator}
  [ A \sqcup A^\prime, B, C]_P = [ A, B, C]_P + [A^\prime, B, C]_P
\end{equation}
and similar additivity properties hold in the other entries.
Using the idempotent property of $P$, we also directly obtain the formula
\begin{equation}
\label{OneOfTheEntriesIsM}
  [A, B, M]_P = [P_A, P_B], 
\end{equation}
where $P_A = PAP$, $P_B = PBP$.

Observe that if $P$ is contained in $\sB_{\rm fin}(M)$ and has propagation bound $r$, then $[ A, B, C]_P$ has propagation bound $3r$, and hence $[ A, B, C]_P=[ A, B, C]_P K $ where
\[
K=A_{3r}\cap B_{3r}\cap C_{3r}.
\]
In particular, if $(A, B, C)$ is a 2-partition, so $K$ is bounded, we have that $[ A, B, C]_P = [ A, B, C]_P K$ is trace class, since $PK$ is, by the defining property of $\sB_{\rm fin}(M)$.

\begin{dfn}\label{dfn:partition.idempotent.pairing}
Let $(A,B,C)$ be a 2-partition of $M$, and $P=P^2\in\sB_{\rm fin}(M)$ be an idempotent. Their \emph{pairing} is
\begin{align}
\label{eqn:Kitaev.2.current}
\langle A,B,C;P\rangle&:=\Tr [A, B, C]_P.
\end{align}
\end{dfn}

\begin{lem}\label{lem:pairing.is.commutator.trace}
For any 2-partition $(A,B,C)$ of $M$, and any $P=P^2\in \sB_{\rm fin}(M)$, the operator $[P_A,P_B]$ is trace class, and
\begin{equation}
\langle A,B,C;P\rangle =\Tr[P_A,P_B].\label{eqn:commutator.trace.class}
\end{equation}
\end{lem}
\begin{proof}
Using the additivity property \eqref{AdditivityGeneralizedCommutator} and skew-symmetry of \eqref{generalizedCommutator}, we calculate
\begin{equation*}
[A, B, C]_P 
= [A, B, M]_P - \underbrace{[A, B, A]_P}_{=0} - \underbrace{[A, B, B]_P}_{=0}  = [P_A, P_B],
\end{equation*}
where in the last equation we also used \eqref{OneOfTheEntriesIsM}.
We saw above that the left hand side is trace class for a 2-partition $(A, B, C)$, hence the same is true for the right hand side.
Taking the trace gives \eqref{eqn:commutator.trace.class}.
\end{proof}

\begin{rem}
For $M=\ZZ^2$, the expressions \eqref{eqn:Kitaev.2.current}-\eqref{eqn:commutator.trace.class} appear in \cite[Eq.\ (125)]{Kitaev}. In this case, one can write $P$ as an infinite matrix, $P = (P_{ij})_{i, j \in \ZZ^2}$, and obtains
\begin{equation}
  \langle A,B,C;P\rangle = 3 \sum_{i \in A}\sum_{j \in B}\sum_{k \in C} (P_{ij} P_{jk} P_{ki} - P_{ik} P_{kj} P_{ji}).\label{eqn:2-current.expression}
\end{equation}
In other words, $\langle A,B,C;P\rangle$ is a signed sum over all oriented triangles such that one vertex lies in $A$, one in $B$ and one in $C$. 
The summand in \eqref{eqn:2-current.expression} was called a ``2-current'' in \cite{Kitaev}.
\end{rem}

\begin{lem}\label{lem:vanishing.pair.coarsely.transverse}
Let $P=P^2\in\sB_{\rm fin}(M)$. Suppose $A,B\subseteq M$ are coarsely transverse, then $\Tr[P_A,P_B]=0$. In particular, if any two of the subsets in a 2-partition $(A,B,C)$ are already coarsely transverse, then $\langle A,B,C;P\rangle=0$.
\end{lem}
\begin{proof}
Similar to the argument given before Definition \eqref{dfn:partition.idempotent.pairing}, one observes that finite propagation of $P$ and coarse transversality of $A,B$ imply that 
$P_AP_B=PAPBP$ and $P_BP_A=PBPAP$ are individually trace class. 
Furthermore, they have the same non-zero eigenvalues. Lidskii's trace theorem \cite{Lidskii} (trace equals the sum of eigenvalues) therefore implies that $0=\Tr(P_AP_B)-\Tr(P_BP_A)=\Tr[P_A,P_B]$.
\end{proof}

\subsection{Cobordism invariance of pairing}
\begin{lem}
\label{CobordismLemma}
Fix a $q$-partition $(A_0, \dots, A_q)$ of $M$.
For indices $i \neq j$, let $W \subseteq A_i$ be a subset such that the collection of sets
\begin{equation}
\label{change.partition.condition}
  W, A_0, \dots, \hat{A}_i, \dots, \hat{A}_j, \dots, A_q  \qquad \text{is coarsely transverse},
\end{equation}
 where the hats denote that $A_i$ and $A_j$ are left out.
Then \mbox{$(A_0, \dots, {A_i \setminus W}, \dots, {W\sqcup A_j},  \dots, A_q)$} is another $q$-partition of $M$.
\end{lem}

\begin{proof}
For all $r>0$, we have
\begin{align*}
(A_i\setminus W)_r\cap (W\sqcup A_j)_r
&=(A_i\setminus W)_r\cap\big[W_r\cup(A_j)_r\bigr]\\
&\subseteq \big[(A_i\setminus W)_r\cap W_r\big]\cup\big[(A_i\setminus W)_r\cap (A_j)_r\big]\\
&\subseteq W_r\cup\big[(A_i)_r\cap(A_j)_r\big].
\end{align*}
Taking the intersection with $\bigcap_{k \notin\{i, j\}} (A_k)_r$, this is the union of two terms, where the first is bounded by \eqref{change.partition.condition} while the second one is bounded since $(A_0, \dots, A_q)$ is a $q$-partition.
\end{proof}

We call two $q$-partitions \emph{elementary cobordant} if one is obtained from the other by taking a portion of one subset to another as in Lemma~\ref{CobordismLemma}, and \emph{cobordant} if one is obtained from the other by a finite sequence of elementary cobordisms.
This terminology is borrowed from \cite{Higson}, which deals with the case of 1-partitions (of a manifold).

\begin{prop}\label{prop:cobordism}
Let $(A, B, C)$ and $(A^\prime, B^\prime, C^\prime)$ be cobordant 2-partitions of $M$.
Then their pairings with any idempotent $P=P^2\in\sB_{\rm fin}(M)$ coincide,
\begin{equation*}
\langle A, B, C; P\rangle = \langle A^\prime, B^\prime, C^\prime; P\rangle.
\end{equation*}
\end{prop}

\begin{proof}
It suffices to consider the case of the elementary cobordism $A^\prime = A \setminus W$, $B^\prime = W\sqcup B$, $C^\prime = C$, where $W \subseteq A$ is such that $W$ and $C$ are coarsely transverse.
For all $P=P^2\in\sB_{\rm fin}(M)$, we now calculate using \eqref{AdditivityGeneralizedCommutator} and \eqref{OneOfTheEntriesIsM}
\begin{equation}
\label{eqn:cobordism.difference.vanish}
\begin{aligned}
[A, B, C]_P - [A^\prime, B^\prime, C^\prime]_P 
&= [A^\prime \sqcup W, B, C]_P - [A^\prime, B \sqcup W, C]_P \\
&= [A^\prime, B, C]_P + [W, B, C]_P - [A^\prime, B, C]_P - [A^\prime, W, C]_P\\
&= [W, B, C]_P + [W, A^\prime, C]_P + \underbrace{[W, W\sqcup C, C]_P}_{=0}\\
&= [W, M, C]_P\\
&= - [P_W, P_C].
\end{aligned}
\end{equation}
Since $W$ and $C$ are coarsely transverse, after passing to traces, we may apply Lemma \ref{lem:vanishing.pair.coarsely.transverse}, to obtain that $\Tr[P_W,P_C]=0$.
\end{proof}

\begin{example}\label{ex:2D.cobordism}
By Lemma \ref{lem:partition.from.multipartition}, coarsely-transverse half-spaces $X,Y$ determine the 2-partitions 
\[
(X, X^c\cap Y, X^c\cap Y^c),\qquad  (Y, X\cap Y^c, X^c\cap Y^c),
\]
with the second one obtained by swapping the roles of $X$ and $Y$. Writing
\begin{align*}
(X, X^c\cap Y, X^c\cap Y^c)&=\bigl((X\cap Y^c)\sqcup \overbrace{(X\cap Y)}^{W}, X^c\cap Y, X^c\cap Y^c\bigr),
\end{align*}
we see that the first partition is cobordant to the second one,
\begin{align*}
 (X\cap Y^c, \underbrace{(X\cap Y)}_{W}\sqcup (X^c\cap Y),X^c\cap Y^c)=(X\cap Y^c,Y,X^c\cap Y^c),
\end{align*}
up to a swap (see Fig.\ \ref{fig:subdivision}). 
Here the condition of Lemma~\ref{CobordismLemma} is indeed satisfied, as
\begin{equation}
W_r\cap(X^c\cap Y^c)_r = (X\cap Y)_r\cap(X^c\cap Y^c)_r \subseteq X_r\cap Y_r\cap (X^c)_r\cap(Y^c)_r\label{eqn:cobordism.example}
\end{equation}
is bounded for all $r>0$, since $X,Y$ are coarsely transverse half-spaces. This says that $W,X^c\cap Y^c$ are coarsely transverse, so $W$ satisfies Condition \eqref{change.partition.condition} in Lemma \ref{CobordismLemma} for implementing this elementary cobordism.
\end{example}

\subsection{Integrality of pairing}\label{sec:integral.pairing.fp}

In this section, we prove the following result.

\begin{thm}
\label{thm:integrality.fp}
Let $(A, B, C)$ be a 2-partition of $M$ and let $P=P^2\in\sB_{\rm fin}(M)$ be an idempotent.
Then
\begin{equation*}
4 \pi i \cdot \langle A, B, C; P\rangle \in \ZZ
\end{equation*}
\end{thm}

We need the following lemma.

\begin{lem}
\label{LemmaTraceHalfSpaces}
Let $X,  Y \subseteq M$ be coarsely transverse half-spaces. Then $[P_X, P_Y]$ is trace class, and 
\begin{equation*}
\frac{1}{2} \Tr[P_X,P_Y]= \langle X,X^c\cap Y,X^c\cap Y^c;P\rangle.
\end{equation*}
\end{lem}

\begin{proof}
We rewrite
\begin{equation}
\label{eqn:subdivided.commutator}
\begin{aligned}
[P_X,P_Y]&=[P_X, P_{X^c\cap Y}]+[P_X, P_{X\cap Y}]\\
&=[P_X, P_{X^c\cap Y}]+[P_{X\cap Y^c}, P_{X\cap Y}]+[P_{X\cap Y}, P_{X\cap Y}]\\
&=[P_X, P_{X^c\cap Y}]-[P_{X\cap Y},P_{X\cap Y^c}]\\
&=[P_X, P_{X^c\cap Y}]-[P_Y,P_{X\cap Y^c}]+[P_{X^c\cap Y},P_{X\cap Y^c}]
\end{aligned}
\end{equation}
For the last commutator, note that
\[
(X^c\cap Y)_r\cap (X\cap Y^c)_r\subseteq (X^c)_r\cap Y_r\cap X_r\cap(Y^c)_r
\]
is bounded since $X,X^c,Y,Y^c$ are coarsely transverse by assumption. So $X\cap Y^c, X^c\cap Y$ are coarsely transverse, and Lemma \ref{lem:vanishing.pair.coarsely.transverse} gives $\Tr[P_{X^c\cap Y},P_{X\cap Y^c}]=0$.
Also, Lemma \ref{lem:partition.from.multipartition} says that 
\begin{equation}
(X,X^c\cap Y,X^c\cap Y^c), \qquad \text{and} \qquad \qquad(Y,X\cap Y^c,X^c\cap Y^c).\label{eqn:possible.partitions}
\end{equation}
are 2-partitions. By Lemma \ref{lem:pairing.is.commutator.trace}, $[P_X, P_{X^c\cap Y}]$ and $[P_Y,P_{X\cap Y^c}]$ are trace class, with
\begin{align*}
\Tr[P_X, P_{X^c\cap Y}]&=\langle X,X^c\cap Y,X^c\cap Y^c;P\rangle\\
\Tr[P_Y,P_{X\cap Y^c}]&=\langle Y,X\cap Y^c,X^c\cap Y^c;P\rangle.
\end{align*}
Thus $[P_X,P_Y]$ is trace class.

In fact, Example \ref{ex:2D.cobordism} showed how the 2-partitions in Eq.\ \eqref{eqn:possible.partitions} are cobordant, up to a swap. 
So, up to a sign, they have equal pairings with $P$ by Prop.\ \ref{prop:cobordism}. 
Thus
\begin{equation*}
\begin{aligned}
\Tr[P_X,P_Y]&=\Tr[P_X, P_{X^c\cap Y}]-\Tr[P_Y,P_{X\cap Y^c}]\\
&=\langle X,X^c\cap Y,X^c\cap Y^c;P\rangle-\langle Y,X\cap Y^c,X^c\cap Y^c;P\rangle\\
&=2\cdot\langle X,X^c\cap Y,X^c\cap Y^c;P\rangle,
\end{aligned}
\end{equation*}
as desired.
\end{proof}

\begin{proof}[Proof of Thm.\ \ref{thm:integrality.fp}]
Let $X$, $Y$ be the coarsely transverse half-spaces corresponding to the 2-partition $(A, B, C)$, as constructed in Lemma \ref{lem:partition.to.multiplatition.rd}.
Then with a view on Remark~\ref{RemarkBackAndForth}, Lemma~\ref{LemmaTraceHalfSpaces} implies that
\begin{equation*}
  4 \pi i \cdot \langle A, B, C; P \rangle = 2 \pi i \cdot \Tr[P_X, P_Y].
\end{equation*}
Because of Prop.\ \ref{prop:minimal.quantization}, it now suffices to verify that the trace class conditions, Eq.\ \eqref{eqn:minimal.tc.assumptions}, are satisfied, in order to deduce integrality of $2\pi i\cdot\Tr[P_X,P_Y]$.
We already know that $[P_X, P_Y]$ is trace class. It remains to check the trace class  properties of
\begin{align*}
(P_X-P_X^2)(P_Y-P_Y^2)
&= (PXP - PXPXP)(PYP-PYPYP) \\
&=\bigl(PXP(1-X)P\bigr)\bigl(PYP(1-Y)P\bigr)\\
&= PXPX^cPYPY^cP
\end{align*}
and
\begin{equation*}
(P_X-P_X^2)^*(P_Y-P_Y^2)=P^*X^cP^*XP^*PYPY^cP.
\end{equation*}
(The products in reverse order are dealt with in the same way.)
Let $r$ be a propagation bound of $P$, thus also of $P^*$, then the above operator products have propagation at most $6r$ and  are supported within
\[
 X_{6r}\cap(X^c)_{6r}\cap Y_{6r}\cap (Y^c)_{6r},
\]
which is a bounded set as $X,Y$ are coarsely transverse. 
Since $P$ is locally trace class, it follows that $(P_X-P_X^2)(P_Y-P_Y^2)$ and $(P_X-P_X^2)^*(P_Y-P_Y^2)$ are trace class.
This finishes the proof.
\end{proof}

The above proof in fact shows the following result.

\begin{cor}\label{cor:quantization.fp.half.space}
Let $X, Y$ be coarsely transverse half-spaces and let $P=P^2\in\sB_{\rm fin}(M)$ be  an idempotent.
Then
\begin{equation*}
2 \pi i \cdot \Tr[P_X, P_Y] \in \ZZ.
\end{equation*}
\end{cor}

\section{Extending pairing to rapid decrease idempotents}\label{sec:extend.pairing}

The fact that the results of the previous section are restricted to idempotents in the algebra $\sB_{\rm fin}(M)$ of finite propagation operators is a serious drawback, as $K$-theoretically non-trivial idempotents that one encounters in nature typically do not have finite propagation (for example, the Landau level spectral projections, see \cite{ASS2, EGS,LT-hyp}).
While the (algebraic) $K$-theory group of the algebra $\sB_{\rm fin}(M)$ is non-trivial, interesting elements typically come from index theory and are produced using some boundary map in (algebraic) $K$-theory. 
These $K$-theory elements are represented by formal differences of idempotents in matrix algebras over the \emph{unitization} of $\sB_{\rm fin}(M)$, and not in $\sB_{\rm fin}(M)$ itself.
In fact, it seems likely that $\sB_{\rm fin}(M)$ does not contain any $K$-theoretically non-trivial idempotents at all.

These problems may have been already apparent to Kitaev, who informally formulated his partition pairing for ``quasi-diagonal'' projections, which do not necessarily have finite propagation but whose kernels decay sufficiently away from the diagonal.

\medskip

Observe that the algebra $\sB_{\rm fin}(M)$
contains, for any subset $Z \subseteq M$, a $*$-ideal $\sB_{\rm fin}(M; Z)$ containing those operators $L$ that are \emph{supported near} $Z$, meaning that there exists $r>0$ such that $LA = AL = 0$ whenever $A\subseteq M$ satisfies $d(A, Z) > r$.
In particular, $\sB(M; M) = \sB(M)$.
The collection of $*$-algebras $\sB(M; Z)$, $Z \subseteq M$, has the following abstract properties:
\begin{enumerate}[(i)]
\item For all Borel subsets $Z \subseteq M$ and $L \in \sB_{\rm fin}(M)$, the operators $ZL$ and $LZ$ are contained in $\sB_{\rm fin}(M; Z)$.
\item If $K \subseteq M$ is a bounded Borel subset, then every operator in $\sB_{\rm fin}(M; K)$ is trace class. 
\item If $Z_0, \dots, Z_q \subseteq M$ is a coarsely excisive collection of subsets of $M$, then
\begin{equation*}
  \sB_{\rm fin}(M; Z_0) \cdot \sB_{\rm fin}(M; Z_1) \cdots \sB_{\rm fin}(M; Z_q) \subseteq \sB_{\rm fin}(M; \textstyle\bigcap_{n=0}^qZ_n).
\end{equation*}
\end{enumerate}
  Here a collection of subsets $Z_0, \dots, Z_q$ is \emph{coarsely excisive} if there exists a non-decreasing function $f: [0, \infty) \to [0, \infty)$ such that 
  \begin{equation}
  \label{CoarselyExcisive}
  \bigcap_{n=0}^q (Z_n)_r \subseteq \Big(\bigcap_{n=0}^q Z_n\Big)_{f(r)} \qquad \forall r \geq 0.
\end{equation}
Inspecting the proofs from the previous section, it is easy to check that all results obtained there hold for $\sB_{\rm fin}(M)$ and $\sB_{\rm fin}(M; Z)$ replaced by any abstract $*$-subalgebra $\sB(M)$ of the algebra of bounded operators on $L^2(M)$, together with a collection of $*$-ideals $\sB(M; Z) \subseteq \sB(M)$ for subsets $Z \subseteq M$, which satisfy (i)-(iii). 
The $*$-closure property is needed only for the integrality results of Subsection \ref{sec:integral.pairing.fp}, where the adjoint idempotent $P^*$ is required to be in $\sB_{\rm fin}(M)$ when the idempotent $P$ is.
We therefore seek to find such  $*$-subalgebras $\sB(M; Z)$ where the finite propagation condition is replaced by a suitable decay condition.

\medskip

Our main observation in this context is that condition (iii) has to be relaxed when allowing more general operators with possibly infinite propagation. 
Intuitively, condition (iii) is problematic if the function $f$ in \eqref{CoarselyExcisive} grows faster than the off-diagonal decrease rate of the kernels of operators in $\sB(M;Z)$.
In the context of faster-than-polynomial decay studied in this paper, it turns out that one should only allow those collections of subsets where $f$ can be chosen to be a polynomial function in Eq.~\eqref{CoarselyExcisive}.
We are therefore led to the following notion.

\begin{dfn}
\label{dfn:poly.coarse.subsets}
We say that a collection of sets $Z_0,\ldots,Z_q\subseteq M$ is \emph{polynomially excisive} if there exists some $\mu\geq 1$ such that for all $r \geq 0$, 
\begin{equation*}
\bigcap_{n=0}^{q}(Z_{n})_r\subseteq\Big(\bigcap_{n=0}^q Z_n\Big)_{r^\mu}.
\end{equation*}
\end{dfn}

\begin{example}
Any thickening of the partitions from Example~\ref{ex:Euclidean.simplex} is polynomially excisive.
In $\RR^2$, the subsets $Z_0 = \{0\}\times \RR$ and $Z_1 = \{(x, \exp(x)) : x \in \RR\}$ do not have polynomially excisive thickenings.
\end{example}

The operator algebras we consider below satisfy  the analogue of the property (iii) for polynomially excisive collections of subsets.
This allows us to derive integrality results along the lines of Section \ref{sec:integral.pairing.fp}, for idempotents in such operator algebras and half-spaces satisfying a polynomial excisiveness condition.

\subsection{Algebras of rapid decrease}
In this Subsection, we define certain operator algebras consisting of operators with kernels having rapid off-diagonal decrease.
It turns out that the correct setting for this construction is the following.

\begin{dfn}
Let $M$ be a proper metric measure space.
\begin{enumerate}
\item[(1)] We say that $M$ is of \emph{bounded geometry} if there exists a countable discrete subset $\Gamma \subseteq M$ such that its $r$-thickening $\Gamma_r$ is all of $M$ for some $r>0$, and such that for each $r>0$, we have
\begin{equation*}
    \sup_{\substack{K \subseteq M \\ \diam(K) \leq r}} \# (K \cap \Gamma) < \infty.
\end{equation*}
\item[(2)] We say that $M$ is of \emph{polynomial volume growth} if there exist constants $C, \mu \geq 0$ such that for all $x \in M$,
\begin{equation*}
  \vol(\{x\}_r) \leq C (1+r)^\mu \qquad \forall r \geq 0.
\end{equation*}
\end{enumerate}
\end{dfn}

For the rest of this subsection, we assume that $M$ has bounded geometry and polynomial volume growth.

\begin{dfn}\label{dfn:ThreeBarNorm.algebras}
For any subset $Z \subseteq M$, we define $\sB(M; Z)$ to be the set of all bounded operators $L$ on $L^2(M)$ for which the quantities
\begin{align}
\label{ThreeBarNorm1}
  \threebar L\threebar_\nu &:= \sup_{\substack{\diam(K) \leq 1 \\\diam(K^\prime)\leq 1}}\|KLK^\prime\|_{\Tr} \bigl(1 + d(K, K^\prime)\bigr)^\nu \\
\label{ThreeBarNorm2}
  \threebar L\threebar_{\nu, Z} &:= \sup_{\substack{\diam(K) \leq 1 \\\diam(K^\prime)\leq 1}} \|KLK^\prime\|_{\Tr} \bigl(1 + d(K, Z)\bigr)^\nu\bigl(1 + d(K^\prime, Z)\bigr)^\nu,
\end{align}
are finite, for any $\nu \geq 0$.
Here the supremum is taken over all sets $K, K^\prime \subseteq M$ of diameter not exceeding one.
We also write $\sB(M) := \sB(M; M)$.
\end{dfn}

It follows directly from the definition that the $\sB(M;Z)$ are $*$-closed, and it is easy to check that $\sB(M; Z) = \sB(M; Z_r)$ for each $r>0$.

\begin{example}\label{ex:Schwartz}
Let $M=\RR^d$, and $L$ a smooth integral kernel operator. 
Restricting $L$ to a bounded region, we get a trace class operator (compare \cite{Roe-partition}, Lemma 1.4, for more general $M$). 
If the integral kernel of $L$ and all its derivatives decay sufficiently rapidly away from the diagonal, uniformly over $M$, these local trace norms decay as in Eq.\ \eqref{ThreeBarNorm1}, so that $L$ lies in $\sB(M)$.
\end{example}

\begin{rem}
For $M=\ZZ^d$, and $Z$ a standard coordinate axis, operator spaces like those in Definition \ref{dfn:ThreeBarNorm.algebras} have been considered before, e.g., the \emph{local} and/or \emph{confined} operators of \S3 in \cite{ST}. 
\end{rem}

Using the assumption of bounded geometry and polynomial volume growth, one may show that $\sB(M)$ is an algebra and each $\sB(M; Z)$ is an ideal in $\sB(M)$.
Moreover, the seminorms \eqref{ThreeBarNorm1}-\eqref{ThreeBarNorm2} turn $\sB(M; Z)$ into a Fr\'echet algebra, which is in fact $m$-convex, hence closed under holomorphic functional calculus.
Taking  these properties for granted for now, the crucial properties for this paper are the following.

\begin{thm}
\label{ThmAbstractPropertiesBMZ}
The family of $*$-algebras $\sB(M; Z)$, $Z \subseteq M$, has the following properties.
\begin{enumerate}[(i)]
\item For all Borel subsets $Z \subseteq M$ and $L \in \sB(M)$, the operators $ZL$ and $LZ$ are contained in $\sB(M; Z)$.
\item If $K \subseteq M$ is a bounded Borel subset, then every operator in $\sB(M; K)$ is trace class. 
\item If $Z_0, \dots, Z_q \subseteq M$ is a polynomially excisive collection of subsets of $M$, then
\begin{equation*}
  \sB(M; Z_0) \cdot \sB(M; Z_1) \cdots \sB(M; Z_q) \subseteq \sB\bigl(M; \textstyle\bigcap_{n=0}^qZ_n\bigr).
\end{equation*}
\end{enumerate}
\end{thm}

\begin{rem}
 Properties (i)--(ii) of Theorem \ref{ThmAbstractPropertiesBMZ} imply that operators in $\sB(M)$ are locally trace class. Despite this, it is not generally true that $\sB_{\rm fin}(M)\subseteq \sB(M)$, because the locally trace class condition does not \emph{uniformly} bound the local trace norms, in the sense of \eqref{eqn:sum.seminorm}. 
\end{rem}

\subsection{Verification of the properties}

Let $M$ be a proper metric measure space of polynomial growth and bounded geometry.
For verification of the claimed properties of the family $\sB(M; Z)$ of operator algebras above, it is convenient to fix a \emph{tiling} of $M$, by which we mean a countable collection $\mathcal{V}$ of Borel subsets, such that
\begin{enumerate}[(a)]
\item $V \cap W = \emptyset$ for all $V, W \in \mathcal{V}$;
\item $\bigcup_{V \in \mathcal{V}} V = M$; 
\item for some $r_0 \geq 0$, we have $\mathrm{diam}(V) \leq r_0$ for all $V \in \mathcal{V}$;
\item $\mathcal{V}$ is \emph{uniformly locally finite}, meaning that for each $r>0$, we have
\begin{equation*}
    \sup_{\substack{K \subseteq M \\ \diam(K) \leq r}} \# \{V \in \mathcal{V} \mid K \cap V \neq \emptyset\} < \infty.
\end{equation*}
\end{enumerate}

\begin{rem}
The existence of such a tiling follows easily from the assumption of bounded geometry.
To construct one, let $\Gamma = \{\gamma_1, \gamma_2, \dots\}$ be a countable set such that $\Gamma_{r_0} = M$.
Then let $V_n^\prime$ be the ball of radius $r_0$ around $\gamma_n$ and set $V_n = V_n^\prime \setminus \{V_1^\prime, \dots, V_{n-1}^\prime\}$.
Then the collection $\mathcal{V} = \{V_1, V_2, \dots \}$ satisfies (i)-(iv).
\end{rem}

For the rest of this subsection, we fix a tiling $\mathcal{V}$.

\begin{lem}
For any subset $Z \subseteq M$, the seminorms Eq.\ \eqref{ThreeBarNorm1}-\eqref{ThreeBarNorm2} are equivalent to the family of seminorms
\begin{align}
\label{BracketSeminorm1}
  [L]_\nu &:= \sup_{V, W \in \mathcal{V}}\|VLW\|_{\Tr} \bigl(1 + d(V, W)\bigr)^\nu \\
\label{BracketSeminorm2}
  [L]_{\nu, Z} &:= \sup_{V, W \in \mathcal{V}}\|VLW\|_{\Tr} \bigl(1 + d(V, Z)\bigr)^\nu\bigl(1 + d(W, Z)\bigr)^\nu.
 \end{align}
\end{lem}

\begin{proof}
Clearly, we have $[L]_\nu \leq \threebar L\threebar_\nu$ and $[L]_{\nu, Z} \leq\threebar L\threebar_{\nu, Z}$.
To show the converse, let $K, K^\prime \subseteq M$ be two Borel subsets of diameter not exceeding one.
First we observe that since the members of $\mathcal{V}$ have diameter at most $r_0$, whenever $V, W \in \mathcal{V}$ are such that $K \cap V \neq \emptyset$ and $K^\prime \cap W \neq \emptyset$, then
\begin{equation*}
  d(K, K^\prime) \leq d(V, W) +2r_0.
\end{equation*}
Hence
\begin{equation*}
\begin{aligned}
 \|KLK^\prime\|_{\Tr} \bigl(1 + d(K, K^\prime)\bigr)^\nu
  &\leq 
  \sum_{\substack{V \in \mathcal{V} \\ K \cap V \neq \emptyset}}
  \sum_{\substack{W \in \mathcal{V} \\ K^\prime \cap W \neq \emptyset}}
  \|VLW\|_{\Tr} \bigl(1 + 2 r_0 + d(V, W) \bigr)^\nu\\
  &\leq (1+2r_0)^\nu \cdot N^2 \cdot [L]_\nu,
\end{aligned}
\end{equation*}
where $N$ is a constant such that $\#\{V \in \mathcal{V} \mid V \cap K \neq \emptyset\} \leq N$ for each subset $K \subseteq M$ with $\diam(K) \leq 1$.
Taking the supremum over $K$ and $K^\prime$ yields the desired estimate.
The estimate for the $\threebar \cdot \threebar_{\nu, Z}$ norms is similar.
\end{proof}

\begin{lem}
For any subset $Z \subseteq M$ the seminorms Eq.\ \eqref{BracketSeminorm1}-\eqref{BracketSeminorm2} are equivalent to the family of seminorms
\begin{align}
\|L\|_\nu&:=\sup_{V \in \mathcal{V}}\sum_{W \in \mathcal{V}}\|VLW\|_\Tr\bigl(1+d(V,W)\bigr)^{\nu}\label{eqn:sum.seminorm} \\
\label{eqn:sum.seminorm.localized}
\|L\|_{\nu,Z}&:=\sup_{V\in \mathcal{V}}\sum_{W \in \mathcal{V}}\|VLW\|_\Tr\bigl(1+d(V,Z)\bigr)^{\nu}\bigl(1+d(W,Z)\bigr)^{\nu}.
\end{align}
\end{lem}

\begin{proof}
The estimates $[L]_\nu \leq \|L\|_\nu$ and $[L]_{\nu, Z} \leq \|L\|_{\nu, Z}$ are obvious.
For the converse estimate, we get 
\begin{equation*}
\begin{aligned}
  \|L\|_\nu 
  &= \sup_{V \in \mathcal{V}} \sum_{W \in \mathcal{V}} \|VLW\|_{\Tr} \bigl(1+d(V, W)\bigr)^\nu\\
  &=  \sup_{V \in \mathcal{V}} \sum_{W \in \mathcal{V}} \bigl(1+d(V, W)\bigr)^{-\mu} \|VLW\|_{\Tr} \bigl(1+d(V, W)\bigr)^{\nu+\mu}\\
  &\leq [L]_{\nu + \mu} \cdot \sup_{V \in \mathcal{V}} \sum_{W \in \mathcal{V}} \bigl(1+d(V, W)\bigr)^{-\mu}.
\end{aligned}
\end{equation*}
This is finite for $\mu$ sufficiently large, by the polynomial growth assumption.
For any subset $Z \subseteq M$, we have
\begin{equation*}
\begin{aligned}
  \|L\|_{\nu, Z} 
  &= \sup_{V \in \mathcal{V}} \sum_{W \in \mathcal{V}} \|VLW\|_{\Tr} \bigl(1+d(V, Z)\bigr)^\nu\bigl(1+d(W, Z)\bigr)^\nu\\
  &\leq   \sup_{V \in \mathcal{V}} \sum_{W \in \mathcal{V}} \bigl(1+d(V, W)\bigr)^{-\mu} \|VLW\|_{\Tr} \bigl(1+d(V, W)\bigr)^{\mu}\bigl(1+d(V, Z)\bigr)^\nu\bigl(1+d(W, Z)\bigr)^\nu\\
  &\leq [L]_{2\nu, Z}^{1/2} \cdot [L]_{2\mu}^{1/2} \cdot  \sup_{V \in \mathcal{V}} \sum_{W \in \mathcal{V}} \bigl(1+d(V, W)\bigr)^{-\mu},
\end{aligned}
\end{equation*}
which is again finite for $\mu$ sufficiently large, by the polynomial growth assumption.
\end{proof}

The seminorms Eq.\ \eqref{eqn:sum.seminorm}-\eqref{eqn:sum.seminorm.localized} are especially useful to verify  the algebra and ideal properties of $\sB(M; Z)$.

\begin{lem}
\label{LemmaBMnormssubmultiplicative}
Up to a scaling, the quantities Eq.\ \eqref{eqn:sum.seminorm} are submultiplicative seminorms on $\sB(M)$.
Hence $\sB(M)$ is an ($m$-convex) Fr\'{e}chet algebra.
\end{lem}

\begin{proof}
We start with the following preliminary considerations.
Let $V, W, U \in \mathcal{V}$. 
Then for any $u, u^\prime \in M$, we have
\begin{equation*}
\begin{aligned}
d(V, W)
= \inf_{v \in V} \inf_{ w\in W} d(v, w) 
&\leq \inf_{v \in V} \inf_{ w\in W} \bigl(d(v, u) + d(u, u^\prime) + d(u^\prime, w)\bigr)\\
&= d(u, V) + d(u, u^\prime) + d(u^\prime, W) \\
&\leq d(u, V) + \diam(U) + d(u^\prime, W).
\end{aligned}
\end{equation*}
Taking the infimum over all $u, u^\prime \in U$ and using that $\diam(U) \leq r_0$, we get for any $\nu \geq 0$ the estimate
\begin{equation}
\label{TriangleII}
\begin{aligned}
  \bigl(1 + d(V, W)\bigr)^\nu &\leq \bigl(1 + r_0 + d(V, U)+ d(U, W)\bigr)^\nu\\
  &\leq (1+r_0)^\nu \bigl(1 + d(V, U))^\nu(1 + d(U, W)\bigr)^\nu.
\end{aligned}
\end{equation}
Hence for $L,L^\prime\in\sB(M)$, we have
\begin{align*}
\|LL^\prime\|_\nu &\leq 
\sup_{V \in \mathcal{V}} \sum_{W, U \in \mathcal{V}}\|VLU\|_\Tr\|UL^\prime W\|_\Tr\bigl(1+d(V,W)\bigr)^\nu \\
&\leq (1+r_0)^\nu \sup_{V \in \mathcal{V}}\sum_{U \in \mathcal{V}}\|VLU\|_\Tr\bigl(1+d(V,U)\bigr)^\nu\sum_{W \in \mathcal{V}}\|UL^\prime W\|_\Tr\bigl(1+d(U,W)\bigr)^\nu \\
&\leq  (1+r_0)^\nu \|L\|_\nu\|L^\prime\|_\nu
\end{align*}
Note that this means, in particular, that $LL^\prime$ always has finite $\nu$-seminorm, so $\sB(M)$ is closed under composition. 
\end{proof}

\begin{lem}
\label{LemmaIdeal}
$\sB(M; Z)$ is an ideal in $\sB(M)$. In fact, for $L \in \sB(M)$ and $L^\prime \in \sB(M;Z)$, we have the estimates
\begin{equation*}
  \|L L^\prime\|_{\nu, Z} \leq \|L\|_\nu \|L^\prime\|_{\nu, Z}, \qquad \text{and} \qquad \|L^\prime L\|_{\nu, Z} \leq \|L\|_\nu \|L^\prime\|_{\nu, Z}.\qquad
\end{equation*}
\end{lem}

\begin{proof}
Similar to Eq.\ \eqref{TriangleII}, for any $U, V \in \mathcal{V}$, we have
\begin{equation*}
  \bigl(1 + d(V, Z)\bigr)^\nu 
  \leq (1+r_0)^\nu \bigl(1 + d(V, U)\bigr)^\nu\bigl(1 + d(U, Z)\bigr)^\nu.
\end{equation*}
Hence for $L \in \sB(M)$ and $L^\prime \in \sB(M;Z)$, we calculate
\begin{equation*}
\begin{aligned}
\|LL^\prime\|_{\nu, Z} 
&\leq 
\sup_{V \in \mathcal{V}}\sum_{U,W \in \mathcal{V}} \|VLU\|_{\Tr}\|UL^\prime W\|_{\Tr}\bigl(1+ d(V, Z)\bigr)^\nu\bigl(1+ d(W, Z)\bigr)^\nu\\
&\leq  (1+r_0)^\nu\sup_{V \in \mathcal{V}} \sum_{U \in \mathcal{V}} \|VLU\|_{\Tr}(1+ d(V,U))^\nu \times \\
&\qquad\qquad\qquad\qquad\qquad\times\sum_{W \in \mathcal{V}} \|UL^\prime W\|_{\Tr}\bigl(1+ d(U, Z)\bigr)^\nu\bigl(1+ d(W, Z)\bigr)^\nu\\
& \leq  (1+r_0)^\nu\|L\|_\nu \|L^\prime\|_{\nu, Z}.
\end{aligned}
\end{equation*}
The other estimate is similar.
%
\end{proof}

The above discussion finishes the proof that $\sB(M)$ is a Fr\'echet algebra and that $\sB(M; Z)$ are ideals in $\sB(M)$.
The following lemma verifies property (iii) of Thm.~\ref{ThmAbstractPropertiesBMZ}.

\begin{lem}[Localization]\label{thm:coarse.transverse.ideals}
Let $Z_0,\ldots, Z_q\subseteq M$ be polynomially excisive. Then
\begin{equation*}
\sB(M;Z_0)\cdot\sB(M;Z_1)\cdot\ldots\cdot\sB(M;Z_q)\subseteq \sB(M;\textstyle\bigcap_{n=0}^qZ_n).
\end{equation*}
\end{lem}

\begin{proof}
Set $Y = \bigcap_{n=0}^q Z_n$.
Let $V \in \mathcal{V}$ and set $\rho_n = d(V, Z_n)$, $n=0, \dots, q$.
Then for each $n$ there exists $x_n \in V$ such that $d(x_n, Z_n) \leq \rho_n$.
Since $\mathrm{diam}(V) \leq r_0$, there also exists $x \in V$ such that $d(x, Z_n) \leq \rho_n + r_0$.
Then with $\rho = \max\{\rho_0, \dots, \rho_q\}+r_0$, we have $x \in \bigcap_{n=0}^q (Z_n)_\rho$.
Consequently, since $Z_0, \dots, Z_q$ is polynomially excisive,  there exists $\mu \geq 1$ such that $x \in Y_{\rho^\mu}$.
Hence
\begin{equation*}
  d(V, Y) \leq \rho^\mu \leq \sum_{n=0}^q (\rho_n+r_0)^\mu = \sum_{n=0}^q \bigl(r_0+d(V, Z_n) \bigr)^\mu 
\end{equation*}
Therefore, we obtain that
\begin{equation*}
  1 + d(V, Y) \leq  C\prod_{n=0}^q \bigl(1+d(V, Z_n)\bigr)^\mu
\end{equation*}
for some constant $C >0$. Let $L_n \in \sB(M, Z_n)$ and set $L = L_0 \cdots L_q$. 
Then
\begin{equation*}
\begin{aligned}
 [L]_{\nu, Y} 
 &= \sup_{V, W \in \mathcal{V}} \|VLW\|_{\Tr} \bigl(1 + d(V, Y)\bigr)^\nu\bigl(1 + d(W, Y)\bigr)^\nu\\
 &\leq \sup_{V, W \in \mathcal{V}} \|VLW\|_{\Tr} \cdot C^{2\nu}\prod_{n=0}^q\bigl(1 + d(V, Z_n)\bigr)^{\nu\mu}\bigl(1 + d(W, Z_n)\bigr)^{\nu\mu}\\
  &= C^{2\nu}\sup_{V, W \in \mathcal{V}} \prod_{n=0}^q\|VLW\|_{\Tr}^{1/(q+1)}\bigl(1 + d(V, Z_n)\bigr)^{\nu\mu}\bigl(1 + d(W, Z_n)\bigr)^{\nu\mu}\\
 &\leq C^{2\nu}\prod_{n=0}^q [L]_{(q+1)\nu \mu, Z_n}^{1/(q+1)}.
\end{aligned}
\end{equation*}
This is finite as by Lemma~\ref{LemmaIdeal}, we have $L \in \sB(M, Z_n)$ for each $n=0, \dots, q$.
\end{proof}

The following lemma verifies property (i) of Thm.~\ref{ThmAbstractPropertiesBMZ}.

\begin{lem}
\label{LemmaZLinBMZ}
For any $L \in \sB(M)$ and any Borel subset $Z \subseteq M$, the operators $ZL$ and $LZ$ are contained in $\sB(M; Z)$.
\end{lem}

\begin{proof}
As in \eqref{TriangleII}, for any $V, W \in \mathcal{V}$, we have 
\begin{equation*}
  \bigl(1 + d(W, Z)\bigr)^\nu
  \leq (1+r_0)^\nu \bigl(1 + d(W, V)\bigr)^\nu\bigl(1 + d(V, Z)\bigr)^\nu,
\end{equation*}
so
\begin{equation*}
\begin{aligned}
  \|ZL\|_{\nu, Z} &\leq \sup_{V \cap Z \neq \emptyset} \sum_{W \in \mathcal{V}} \|VLW\|_{\Tr} \bigl(1+ d(V, Z)\bigr)^\nu\bigl(1+ d(W, Z)\bigl)^\nu\\
  &\leq (1+r_0)^\nu\sup_{V \cap Z \neq \emptyset} \sum_{W \in \mathcal{V}} \|VLW\|_{\Tr} \bigl(1+ \underbrace{d(V, Z)}_{0}\bigr)^{2\nu}\bigl(1+ d(W, V)\bigr)^\nu \leq \|L\|_{\nu}.
\end{aligned}
\end{equation*}
Hence $ZL \in \sB(M;Z)$.
For $LZ$, the result follows from the fact that $\sB(M)$ and $\sB(M; Z)$ are $*$-closed.
\end{proof}

\begin{lem}\label{lem:automatic.locally.trace.class}
For any bounded Borel subset $K\subseteq M$, and any $L\in\sB(M)$, the operators $KL$ and $LK$ are  trace class, i.e., $L$ is locally trace class.
\end{lem}

\begin{proof}
Since $K$ is bounded, the number $N$ of sets $V \in \mathcal{V}$ such that $V \cap K \neq \emptyset$ is finite.
Therefore
\begin{equation*}
\|KL\|_\Tr \leq \sum_{V \cap K \neq \emptyset} \|VL\|_{\Tr} \leq \sum_{V \cap K \neq \emptyset} \sum_{W \in \mathcal{V}}\|VLW\|_{\Tr} \leq N \|L\|_0.
\end{equation*}
The result for $LK$ follows since $\sB(M; K)$ and the space of trace class operators are $*$-closed.
\end{proof}

Finally, the following lemma verifies property (ii) of Thm.~\ref{ThmAbstractPropertiesBMZ} and hence completes its proof.

\begin{lem}\label{prop:compact-localized.trace.class}
For any bounded subset $K\subseteq M$, the elements of $\sB(M;K)$ are trace class.
\end{lem}

\begin{proof}
We have
\begin{equation*}
\begin{aligned}
  \|L\|_{\Tr} &\leq \sum_{V \in \mathcal{V}} \sum_{W \in \mathcal{V}}\|VLW\|_{\Tr}\\
  &\leq \sum_{V \in \mathcal{V}} \bigl(1 + d(V, K)\bigr)^{-\mu} \sum_{W \in \mathcal{V}}\|VLW\|_{\Tr}\bigl(1 + d(V, K)\bigr)^\mu\\
  &\leq \sum_{V \in \mathcal{V}} \bigl(1 + d(V, K)\bigr)^{-\mu} \cdot \|L\|_{\mu, K},
\end{aligned}
\end{equation*}
which is finite for $\mu$ sufficiently large as $M$ has polynomial growth.
\end{proof}

\subsection{Integrality of pairing}

We now have the ingredients to improve the pairing, Definition \ref{dfn:partition.idempotent.pairing}, to allow for $P=P^2\in\sB(M)$. 
As in the previous Subsection, let $M$ be a proper metric measure space of bounded geometry and polynomial volume growth.
We start with the following lemma.

\begin{lem}\label{lem:poly.trace.class}
Let $Z_0,\dots,Z_q\subseteq M$ be coarsely transverse and assume that the sets $(Z_0)_r, \dots, (Z_q)_r$ are polynomially excisive for some $r>0$.
Then for any $L_0,\dots ,L_q\in\sB(M)$, the product $Z_0L_0 \cdots Z_qL_q$ is trace class. 
\end{lem}

\begin{proof}
This follows directly from the properties (i)-(iii) of Thm.~\ref{ThmAbstractPropertiesBMZ}: 
By (i), we have  $Z_n L_n \in \sB(M; Z_n)$, $n=0, \ldots, q$.
Then applying (iii) to the polynomially excisive collection $(Z_0)_r, \dots, (Z_q)_r$ we obtain that the product of the $Z_n L_n$ is contained in the subalgebra $\sB(M; \bigcap_{n=0}^q (Z_n)_r)$.
On the other hand, as $Z_0, \dots, Z_n$ are coarsely transverse, the intersection $\bigcap_{n=0}^q (Z_n)_r$ is bounded.
Hence property (ii) applies to give the trace class property of the product.
\end{proof}

\begin{dfn}
A \emph{polynomial $q$-partition} $(A_0,\ldots,A_q)$ of $M$ is a $q$-partition of $M$ such that for some $r>0$, the sets $(A_0)_r, \dots, (A_q)_r$ are polynomially excisive.
\end{dfn}

\begin{prop}\label{prop:poly.partition.idempotent.pairing}
There is a well-defined pairing of $P=P^2\in\sB(M)$ with polynomial 2-partitions,
\[
\langle A,B,C;P\rangle:=\Tr[A,B,C]_P=\Tr[P_A,P_B],
\]
which is continuous in $P$.
\end{prop}

\begin{proof}
By Lemma \ref{lem:poly.trace.class},
the operator $APBPCP$ and its antisymmetrization $[A,B,C]_P$ are trace class. 
The equality
$
\Tr[A,B,C]_P=\Tr[P_A,P_B]
$
follows from the same calculation as Lemma \ref{lem:pairing.is.commutator.trace}. 
Continuity follows from the seminorm bounds used to establish Lemma \ref{thm:coarse.transverse.ideals} and Lemma \ref{LemmaIdeal}, and the fact that the Fr\'{e}chet topology is stronger than the trace-norm topology.
\end{proof}

\begin{dfn}\label{dfn:poly.coarse}
A pair of coarsely transverse half-spaces $X,Y\subseteq M$ is said to be \emph{polynomially transverse}, if their associated quadrants $X\cap Y, X\cap Y^c, X^c\cap Y, X^c\cap Y^c$ satisfy the extra condition
\begin{enumerate}
\item[$(\star)$]\label{item:quadrant} 
There exists some $r_0>0$ such that any subcollection of
\[
(X\cap Y)_{r_0},\quad (X\cap Y^c)_{r_0},\quad (X^c\cap Y)_{r_0},\quad (X^c\cap Y^c)_{r_0},
\]
is polynomially excisive.
\end{enumerate}
\end{dfn}

\begin{lem}\label{lem:half.spaces.auto.poly}
Let $X,Y\subseteq M$ be subsets such that their associated quadrants satisfy Condition $(\star)$ of Definition \ref{dfn:poly.coarse}. Then some thickening of the sets $X,X^c,Y,Y^c$ is polynomially excisive.
\end{lem}

\begin{proof}
For any $r \geq 0$,
\begin{align*}
&X_r\cap (X^c)_r\cap Y_r\cap (Y^c)_r\\
&\quad=\big[(X\cap Y)_r\cup (X\cap Y^c)_r\big]\cap\big[(X^c\cap Y)_r)\cup(X^c\cap Y^c)_r\big]\\
&\qquad\quad\cap\big[(X\cap Y)_r\cup(X^c\cap Y)_r\big]\cap \big[(X\cap Y^c)_r\cup(X^c\cap Y^c)_r\big].
\end{align*}
expands into a union of either double, triple, or quadruple intersections of thickenings of ``quarter spaces''.
If now $r \geq r_0$ where $r_0$ is the thickening radius from Definition \ref{dfn:poly.coarse}, say $r = r_0 + s$, then
due to polynomial excisiveness, any of the triple intersections satisfies, e.g.,
\begin{align*}
&(X\cap Y)_r\cap(X^c\cap Y)_r\cap (X^c\cap Y^c)_r\\
&\qquad\qquad\qquad \subseteq\big((X\cap Y)_{r_0}\cap(X^c\cap Y)_{r_0}\cap(X^c\cap Y^c)_{r_0}\big)_{s^\mu}\\
&\qquad\qquad\qquad \subseteq \big(X_{r_0}\cap (X^c)_{r_0}\cap Y_{r_0}\cap(Y^c)_{r_0}\big)_{s^\mu},
\end{align*}
for some $\mu \geq 1$; similarly for the other triple and quadruple intersections.
The double intersections are
\[
(X\cap Y)_r\cap (X^c\cap Y^c)_r\;\;\;{\rm and}\;\;\;(X\cap Y^c)_r\cap(X^c\cap Y)_r,
\]
and polynomial excisiveness ensures that they are each contained within
\[
\big(X_{r_0}\cap (X^c)_{r_0}\cap Y_{r_0}\cap(Y^c)_{r_0}\big)_{s^\mu}
\]
for some $\mu \geq 1$.
So we conclude that for all $s>0$,
\[
(X_{r_0})_s\cap((X^c)_{r_0})_s\cap(Y_{r_0})_s\cap((Y^c)_{r_0})_s\subseteq\big(X_{r_0}\cap (X^c)_{r_0}\cap Y_{r_0}\cap(Y^c)_{r_0}\big)_{s^\mu},
\]
meaning that $X_{r_0},(X^c)_{r_0},(Y)_{r_0},(Y^c)_{r_0}$ is polynomially excisive.
\end{proof}

We may obtain polynomial 2-partitions from polynomially transverse half-spaces:

\begin{lem}\label{lem:partition.from.multipartition.rd}
Let $X,Y\subseteq M$ be polynomially transverse half-spaces. Then
\begin{equation*}
(X, X^c\cap Y, X^c\cap Y^c)\qquad \text{and} \qquad (Y, X\cap Y^c, X^c\cap Y^c)
\end{equation*}
are polynomial 2-partitions of $M$.
\end{lem}
\begin{proof}
By assumption, $X,X^c,Y,Y^c$ are coarsely transverse, and Lemma \ref{lem:partition.from.multipartition} says that $(X,X^c\cap Y,X^c\cap Y^c)$ is a 2-partition.

Let $r_0>0$ be the thickening radius in Definition \ref{dfn:poly.coarse}. 
We shall verify that the collection \mbox{$X_{r_0}, (X^c\cap Y)_{r_0}, (X^c\cap Y^c)_{r_0}$} is polynomially excisive. 
By the polynomial excisiveness of the $r_0$-thickened quadrants, there exists $\mu\geq 1$ such that for all $s>0$,
\begin{align*}
&(X_{r_0})_s\cap((X^c\cap Y)_{r_0})_s\cap((X^c\cap Y^c)_{r_0})_s\\
&\quad=\big[((X\cap Y)_{r_0})_s\cup((X\cap Y^c)_{r_0})_s\big]\cap((X^c\cap Y)_{r_0})_s\cap((X^c\cap Y^c)_{r_0})_s\\
&\quad=\big[((X\cap Y)_{r_0})_s\cap((X^c\cap Y)_{r_0})_s\cap((X^c\cap Y^c)_{r_0})_s\big]\\
&\qquad\qquad\qquad \cup\big[((X\cap Y^c)_{r_0})_s\cap ((X^c\cap Y)_{r_0})_s\cap((X^c\cap Y^c)_{r_0})_s\big]\\
&\quad\subseteq\big((X\cap Y)_{r_0}\cap(X^c\cap Y)_{r_0}\cap (X^c\cap Y^c)_{r_0}\big)_{s^\mu}\\
&\qquad\qquad\qquad \cup \big((X\cap Y^c)_{r_0}\cap(X^c\cap Y)_{r_0}\cap (X^c\cap Y^c)_{r_0}\big)_{s^\mu}\\
&\quad=\big((X)_{r_0}\cap(X^c\cap Y)_{r_0}\cap (X^c\cap Y^c)_{r_0}\big)_{s^\mu},
\end{align*}
as required.
Swapping the roles of $X,Y$ gives the result for the other partition. 
\end{proof}

\begin{thm}\label{thm:quantization.RD}
Let $X,Y\subseteq M$ be polynomially transverse half-spaces. 
For any idempotent $P=P^2\in\sB(M)$,
\[
2\pi i\cdot\Tr[P_X,P_Y]=4\pi i\cdot\langle X,X^c\cap Y,X^c\cap Y^c;P\rangle\;\in\;\ZZ,
\]
and is locally constant in $P$.
\end{thm}
\begin{proof}
The proof is similar to that of Theorem \ref{thm:integrality.fp}--Corollary \ref{cor:quantization.fp.half.space}. 

As before, we first note that $P_X-P_X^2=P(XP)(X^cP)$ and $P_Y-P_Y^2= P(YP)(Y^cP)$,
\[
(P_X-P_X^2)(P_Y-P_Y^2)= P(XP)(X^cP)(YP)(Y^cP).
\]
The sets $X,X^c,Y,Y^c$ are coarsely transverse, and Lemma \ref{lem:half.spaces.auto.poly} says that they have polynomially excisive thickenings. 
By Lemma \ref{lem:poly.trace.class},
\begin{align*}
(XP)(X^cP)(YP)(Y^cP)
\end{align*}
is trace class, thus $(P_X-P_X^2)(P_Y-P_Y^2)$ is trace class. 
Likewise for $(P_Y-P_Y^2)(P_X-P_X^2)$.
The trace class property of
\[
(P_X-P_X^2)^*(P_Y-P_Y^2)= P^*(X^cP^*)(XP^*P)(YP)(Y^cP)
\]
also follows from Lemma \ref{lem:poly.trace.class}, since $P^*, P^*P\in\sB(M)$.
Likewise for $(P_Y-P_Y^2)(P_X-P_X^2)^*$.

As for $[P_X,P_Y]$, recall from Eq.\ \eqref{eqn:subdivided.commutator} that
\[
[P_X,P_Y]=[P_X, P_{X^c\cap Y}]-[P_Y, P_{X\cap Y}]+[P_{X^c\cap Y},P_{X\cap Y^c}].
\]
Note that $X^c\cap Y, X\cap Y^c$ are coarsely transverse, and they have polynomially excisive thickenings by Condition $(\star)$ of Definition \ref{dfn:poly.coarse}. 
So the last commutator in the above equation, 
\[
[P_{X^c\cap Y},P_{X\cap Y^c}]=P(X^c\cap Y)P(X\cap Y^c)P-P(X\cap Y^c)P(X^c\cap Y)P,
\]
is trace class (Lemma \ref{lem:poly.trace.class}) and traceless (Lidskii's theorem).
 Also, by Lemma \ref{lem:partition.from.multipartition.rd}, we have polynomial 2-partitions
\begin{equation}
(X,X^c\cap Y,X^c\cap Y^c),\qquad(Y,X\cap Y^c,X^c\cap Y^c),\label{eqn:poly-partitions}
\end{equation}
and Prop.\ \ref{prop:poly.partition.idempotent.pairing} says that
\begin{align*}
\Tr[P_X, P_{X^c\cap Y}]&=\langle X,X^c\cap Y, X^c\cap Y^c;P\rangle\\
\Tr[P_Y, P_{X\cap Y}]&=\langle Y,X\cap Y^c, X^c\cap Y^c;P\rangle
\end{align*}
make sense. Therefore $[P_X,P_Y]$ is trace class as well. 
We may now apply Prop.\ \ref{prop:minimal.quantization} to conclude that $2\pi i\cdot\Tr[P_X,P_Y]\in\ZZ$.

In fact, Example \ref{ex:2D.cobordism} furnishes a cobordism (up to a swap) between the polynomial 2-partitions in Eq.\ \eqref{eqn:poly-partitions}, implemented by $W=X\cap Y$. The sets $W=X\cap Y$ and $C=X^c\cap Y^c$ are coarsely transverse, and have polynomially excisive thickenings (Condition $(\star)$ of Definition \ref{dfn:poly.coarse}), so Lemma \ref{lem:poly.trace.class} and Lidskii's theorem imply $\Tr[P_W,P_C]=0$. 
This means that the cobordism invariance result, Eq.\ \eqref{eqn:cobordism.difference.vanish} of Prop.\ \ref{prop:cobordism}, still applies. 
Thus,
\begin{align*}
\Tr[P_X,P_Y]&=\Tr[P_X, P_{X^c\cap Y}]-\Tr[P_Y, P_{X\cap Y}]\\
&=\langle X,X^c\cap Y, X^c\cap Y^c;P\rangle-\langle Y,X\cap Y^c, X^c\cap Y^c;P\rangle\\
&=2\cdot\langle X,X^c\cap Y, X^c\cap Y^c;P\rangle.
\end{align*}
Continuity of the pairing, Prop.\ \ref{prop:poly.partition.idempotent.pairing}, implies that $\Tr[P_X,P_Y]$ is continuous in $P$, thus locally constant.
\end{proof}

\begin{rem}
If we start with an arbitrary polynomial 2-partition $(A,B,C)$, the bisection construction of Lemma \ref{lem:partition.to.multiplatition.rd} may not result in \emph{polynomially} transverse half-spaces, so Theorem \ref{thm:quantization.RD} may not be applicable. Because of this, we do not know whether \mbox{$4\pi i\cdot\langle A,B,C;P\rangle\in\ZZ$} holds for all  idempotents $P$ in $\sB(M)$, as was the case for $P\in \sB_{\rm fin}(M)$ (Theorem \ref{thm:integrality.fp}).
\end{rem}

\section{Coarse cohomology viewpoint}\label{sec:coarse.cohomology.viewpoint}

The reader familiar with coarse cohomology and index theory, as developed in \cite{Roe-cohom}, may recognize that our constructions and results can be placed in this abstract framework.

\subsection{Pairing with $K$-theory}\label{sec:additivity}
A \emph{coarse $q$-cochain} on $M$, in the sense of \S 2.2 of \cite{Roe-cohom}, is a locally bounded Borel map $\varphi:M^{q+1}\to\RR$, such that $\Supp(\varphi)\cap (\Delta)_r$ is bounded for all $r>0$, where $\Delta$ denotes the diagonal in $M^{q+1}$. 
The coboundary map is
\[
\delta\varphi(x_0,\ldots,x_{q+1})=\sum_{i=0}^{q+1}(-1)^i \varphi(x_0,\ldots,\hat{x_i},\ldots,x_{q+1}),
\]
where $\hat{x_i}$ denotes omission of the $i$-th argument. Restricting to antisymmetric cochains, we obtain a cochain complex, and the \emph{coarse cohomology groups} $HX^\bullet(M)$.

A $q$-partition $(A_0,\ldots,A_q)$ in the sense of our Definition \ref{dfn:partition} determines an antisymmetric coarse $q$-cochain
\[
\varphi_{A_0,\ldots,A_q}:(x_0,\ldots,x_q)\mapsto \sum_{\sigma \in S_{0, \dots, q}} \mathrm{sgn}(\sigma) \chi_{A_{\sigma_0}}(x_0)\ldots \chi_{A_{\sigma_q}}(x_q),
\]
where $S_{0, \dots, q}$ denotes the group of permutations of $\{0, \dots, q\}$.
 Furthermore, $\delta\varphi_{A_0,\ldots,A_q}=0$, since it is antisymmetric in $q+2$ arguments, but depends only on the question of which of the mutually disjoint $A_i, i=0,\ldots,q$ its arguments belong to. 
 Thus, a $q$-partition defines a coarse cohomology class. 
 It is not hard to check that the partition of a Euclidean space described in Example \ref{ex:Euclidean.simplex} determines a non-trivial coarse cohomology class. We may also verify that cobordant $q$-partitions determine the same coarse cohomology class.

As explained in \S 4.2 of \cite{Roe-cohom}, a coarse cohomology class $[\varphi]$ determines a \emph{cyclic cohomology class} $\chi[\varphi]$ on $\sB_{\rm fin}(M)$. (Here $\chi$ denotes Roe's Connes character map, not a characteristic function.) 
In even degree, $\chi[\varphi]$ has a pairing with idempotents in $\sB_{\rm fin}(M)$.
 For example, our partition-idempotent pairing formula, Definition \ref{dfn:partition.idempotent.pairing}, is
\[
\langle A,B,C;P\rangle\equiv \Tr[A,B,C]_P=\chi[\varphi_{A,B,C}](P,P,P),
\]
i.e., the evaluation of the associated cyclic 2-cocycle on $P$.

By standard constructions (II.2 of \cite{Connes}), the pairing is canonically extended to idempotents coming from matrix algebras over the unitization $\sB_{\rm fin}(M)^+$. 
Then one uses a general result that inner automorphisms on an algebra induce the identity morphism on its cyclic cohomology (Prop.\ II.5 of \cite{Connes}).
The upshot is that the pairing descends to the \emph{algebraic} $K_0$-class of $P$.

Let $P_1,P_2$ be orthogonal idempotents, $P_1P_2=0=P_2P_1$. The identity
\[
\begin{pmatrix}P_1+P_2 & 0 \\ 0 & 0 \end{pmatrix}=\underbrace{\begin{pmatrix} P_1 & 1-P_1\\ P_1-1 & P_1\end{pmatrix}}_{u}\begin{pmatrix}P_1 & 0 \\ 0 & P_2\end{pmatrix}\underbrace{\begin{pmatrix} P_1 & P_1-1\\ 1-P_1 & P_1\end{pmatrix}}_{u^{-1}}
\]
shows that $[P_1+P_2]=[P_1]+[P_2]=[P_1\oplus P_2]$ in $K_0(\mathscr{B}_{\rm fin}(M))$, so the pairing is actually additive with respect to \emph{internal} orthogonal sum of idempotents, not just external direct sum.
This additivity is not obvious since $P$ appears three times in each term of the pairing formula. The internal additivity is important because we often utilize functional calculus of a Hamiltonian operator on a fixed Hilbert space.

\subsection{Relation to Roe's index theorems}\label{sec:Roe.index}
In \S4 of \cite{Roe-cohom}, Roe studied the abstract \emph{coarse indices} ${\rm Ind}(D)\in K_0(\sB_{\rm fin}(M))$ of Dirac-type operators $D$ on $M$, see (4.32) of \cite{Roe-cohom}. In his Prop.\ 5.29, the pairing of $\chi[\varphi]$ with ${\rm Ind}(D)$ was identified, via an index theorem (\S4.42 of \cite{Roe-cohom}), with a $K$-homology--$K$-theory pairing on the so-called \emph{Higson corona} $N$ of $M$ (modulo technical assumptions on $[\varphi]$ satisfied by those built from partitions). It follows that $\langle \chi[\varphi];{\rm Ind}(D)\rangle$ is integral, up to an appropriate normalization factor. In particular, for $D$ the spin Dirac operator on $M=\RR^{2m}$, and $\varphi$ a generator of $HX^{2m}(\RR^{2m})$, the result
\begin{equation*}
\langle \chi[\varphi];{\rm Ind}(D)\rangle=\frac{m!}{(2m)!(2\pi i)^m}
\end{equation*}
was calculated in \S4.4 of \cite{Roe-cohom}. For $m=1$, this coincides with the normalization factor $\frac{1}{4\pi i}$ in our Theorem \ref{thm:integrality.fp}. 

\medskip

On the one hand, Roe's results imply integrality for the pairing of certain coarse cohomology classes with the Dirac index class in $K_0(\sB_{\rm fin}(M))$. However, ${\rm Ind}(D)$ is represented as a formal \emph{difference} $[P]-[1_k]$, with $P=P^2=P^*\in{\rm M}_\infty(\sB_{\rm fin}^+(M))$ coming from the \emph{unitization}. 
For projections in unitizations, the pairing formula ignores the scalar part (see pp.\ 276 of \cite{Connes}), and consequently, the derivation of the commutator-trace formula, Eq.\ \eqref{eqn:commutator.trace.class}, may not hold. So Roe's Dirac index class may not provide an example of a projection $P\in{\rm M}_\infty(\sB_{\rm fin}(M))$ having non-trivial $\Tr[P_X,P_Y]$. It is unclear to us whether ${\rm Ind}(D)$ admits a representation as ${\rm Ind}(D)=[P^\prime]$ for some projection $P^\prime$ in ${\rm M}_\infty(\sB_{\rm fin}(M))$.

On the other hand, our Theorem \ref{thm:integrality.fp}--Corollary \ref{cor:quantization.fp.half.space} give the integrality of the quantities \mbox{$2\pi i\cdot\Tr[P_X,P_Y]$} and $4\pi i\cdot[P_A,P_B]$, for \emph{all} idempotents $P$ in (matrix algebras over) $\sB_{\rm fin}(M)$, and \emph{all} 2-partitions $(A,B,C)$ and coarsely transverse half-spaces $X,Y$.

\medskip

 As explained at the beginning of Section \ref{sec:extend.pairing}, it is unclear to us whether non-zero integer pairings can be realized with the finite-propagation idempotents from $\sB_{\rm fin}(M)$, without passing to the unitization. 
 This difficulty is avoided with $\sB(M)$. 
For example, on $M=\RR^2$, the Dirac operator coupled to a magnetic vector potential for uniform magnetic field has coarse index represented by a projection in $\sB(M)$ --- the spectral projection for the lowest Landau level. This is briefly recalled in Section \ref{sec:physics}. So our Theorem \ref{thm:quantization.RD} for $\sB(M)$ is not vacuous.

\section{Physics applications}\label{sec:physics}

\subsection{Hall conductance}\label{sec:Hall.conductance}
For applications to physics, we have $M$ a Riemannian manifold, and a self-adjoint magnetic Schr\"{o}dinger operator $H$ as the Hamiltonian. If $H$ models non-interacting fermions on $M$, then the low energy states of $H$ are successively occupied by the fermions in the sample in order of increasing energy, until the Fermi energy $E\in \RR$ is reached. 

Assuming that $E$ lies in a spectral gap of $H$, we can obtain the spectral projection onto energies below $E$ by functional calculus, $P=P_E=\psi(H)$, with $\psi$ a smooth compactly-supported real-valued function which equals 1 on the part of the spectrum of $H$ below $E$. This $P$ is called a \emph{Fermi projection} at Fermi level $E$. It is generally not finite propagation, but quite often, one may prove that $P$ has rapid decrease (e.g.\ Appendix A of \cite{ASS2}, \cite{LT-cobordism}), so $P=P^*=P^2\in\sB(M)$. Then these Fermi projections have integral pairings $2\pi i\cdot\Tr[P_X,P_Y]$ with polynomially transverse half-spaces, by Theorem \ref{thm:quantization.RD}.

Let us rewrite
\begin{equation}\label{eqn:Kubo.formula}
\begin{aligned}
[P_X,P_Y]&\equiv PXPYP-PYPXP\nonumber\\
&=P(XP-PX)(YP-PY)-P(YP-PY)(XP-PX)\nonumber\\
&=P[[X,P],[Y,P]].
\end{aligned}
\end{equation}
The last expression in Eq.\ \eqref{eqn:Kubo.formula} would be familiar to physicists as a \emph{Kubo formula} for the Hall conductance (see Eq.\ 1.4 of \cite{ES}, Eq.\ 128--129 of \cite{Kitaev}, \S 6 of \cite{ASS2}), typically used when $M=\RR^2$ or $\ZZ^2$, $X$ is the right half-plane, and $Y$ is the upper half-plane. More specifically, one applies an electric potential difference between $X^c$ and $X$, and considers the induced current from $Y^c$ to $Y$, taken in an adiabatic limit so that linear response may be justified (see, e.g., \cite{ES} for details). We mention that in the homogeneous Euclidean geometry, with uniform magnetic field, one sometimes takes $X,Y$ in Eq.\ \eqref{eqn:Kubo.formula} to be, respectively, the multiplication operators by the global $x$ and $y$ coordinate functions. 
Then a trace-per-unit-volume (thermodynamic limit) is used.
 In this simplified geometry, the trace class property and integrality of the trace of Eq.\ \eqref{eqn:Kubo.formula} were investigated in \cite{EGS}, and in \cite{ASS2} using the relative index of projections, as well as in \cite{BES} using noncommutative geometry and cyclic cohomology of the magnetic Brillouin torus. 

The expression $4\pi i\cdot[P_A,P_B]$, with $(A,B,C)$ a (polynomial) 2-partition, also has the interpretation as the response to a magnetic flux localized at the intersection of the (thickened) $A,B,C$, as discussed in the Supplementary Information of \cite{Mitchell} in the context of amorphous topological phases.

Our paper places such quantization results in a vastly more general and rigorous context, and makes their large-scale geometric origin manifest.

\subsection{Localization and generic triviality of pairing}\label{sec:trivial.pairing.localized}

For a large class of projections $P\in\mathscr{B}(M)$, we automatically have trivial $\Tr[P_X,P_Y]$. A non-exhaustive list is:
\begin{enumerate}[(1)]
\item $P$ is finite-rank.
\item $P$ is supported in a bounded set, or more generally, within any half-space $X$. In this case, by choosing $X$ to be slightly larger, we have $P_X\equiv P\chi_{X}P=P$, which commutes with $P_Y=P\chi_YP$.
\item\label{point:error} Either $P_X$ or $P_Y$ differs from a projection only by a trace class operator. This is because the commutator of a projection with any operator has vanishing trace.
\item $M$ is an orientable manifold, or CW-complex, graph, etc, and $P$ is invariant under orientation reversal, which would effect $[\varphi_{A,B,C}]\mapsto -[\varphi_{A,B,C}]$. 
\item As $P$ is self-adjoint, the commutator $[P_X,P_Y]$ is skew-adjoint, with purely imaginary trace. So if $P$ is \emph{real}, then this commutator-trace would vanish.
\end{enumerate}
We note that in the quantum Hall effect, the external magnetic field breaks the parity symmetry of the Schr\"{o}dinger operator, as well as its invariance under complex conjugation (time-reversal), thereby avoiding the mandatory trivialization of the pairing. The idea of Chern insulators is essentially based on models which break these fundamental symmetries without using magnetic fields.

``Generic vanishing'' of the pairing is actually a desirable feature. Additivity (see Section \ref{sec:additivity}) says that a non-trivial pairing is stable against the orthogonal addition/formal subtraction of any trivial $P^\prime$ of the above form. This means that any $P$ possessing a non-trivial pairing with \emph{some} partition must be ``extremely delocalized'' over all of $M$, in the sense of not admitting \emph{any} orthogonal splitting into trivial projections; compare the Wannier localizability problem studied in \cite{LT-wannier}.

\subsection{Nontrivial pairing for Landau level spectral projection}\label{sec:algebra.comparison}
It was explained in \cite{LT-cobordism,LT-hyp,KLT-coarse} that for a Landau Hamiltonian on a rather general complete Riemannian 2-manifold $M$, the spectral projection $P_{\rm Lan}$ for an isolated Landau level is the kernel projection for a twisted Dirac operator $D$, and lies in $\sB(M)$, so it represents the Dirac index class  ${\rm Ind}(D)$ considered in $K_0(\sB(M))$.

Now let $X,Y$ be polynomially transverse half-spaces, and $(A,B,C)$ be the associated polynomial 2-partition (Lemma \ref{lem:partition.from.multipartition.rd}),
\[
A=X,\qquad B=X^c\cap Y, \qquad C=X^c\cap Y^c.
\]
The coarse 2-cochain $\varphi_{A,B,C}$ determines a cyclic cohomology class of $\sB(M)$, not just of $\sB_{\rm fin}(M)$. Including a subscript to clarify where pairings are taken, we have for $P=P_{\rm Lan}$,
\begin{align*}
\tfrac{1}{2}\Tr[P_X,P_Y]=\langle A,B,C;P_{\rm Lan}\rangle_{\sB(M)}&=\langle\chi[\varphi_{A,B,C}];[P_{\rm Lan}]\rangle_{\sB(M)}\\
&=\langle\chi[\varphi_{A,B,C}];{\rm Ind}(D)\rangle_{\sB(M)}\\
&=\langle\chi[\varphi_{A,B,C}];{\rm Ind}(D)\rangle_{\sB_{\rm fin}(M)}.
\end{align*}
As mentioned in Section \ref{sec:Roe.index}, Roe provides a formula for the last pairing in \cite[Theorem 4.22]{Roe-cohom}, which reduces to $\tfrac{1}{4\pi i}$ for the case of $M=\RR^2$ and $[\varphi_{A,B,C}]$ a generator of $HX^2(M)$.

For $M=\RR^2$, the non-triviality of the Hall conductance of a Landau level projection is well-known (e.g.\ \cite{ASS2, BES}), and the above discussion provides an independent verification.

\subsection{Large finite-sized sample}\label{sec:finite.size}
An actual physical sample $M$ is bounded. 
Let us consider $M$ as a compact subspace inside a larger unbounded $\tilde{M}$. 
Consider on $\tilde{M}$, a fictitious gapped Hamiltonian $\tilde{H}$ (e.g.\ Landau Hamiltonian), with a Fermi projection $\tilde{P}\in\sB(\tilde{M})$ whose Hall conductance is, by Theorem \ref{thm:quantization.RD}, necessarily exactly quantized, 
\[
\sigma_{\tilde{M}}(\tilde{P}):=4\pi i\cdot\Tr[A,B,C]_{\tilde{P}}=4\pi i\cdot\Tr[\tilde{P}_A,\tilde{P}_B]\in \ZZ.
\]
As by Theorem \ref{ThmAbstractPropertiesBMZ}, $[A, B, C]_P$ decays rapidly away from $A_{r_0} \cap B_{r_0} \cap C_{r_0}$ for some sufficiently large $r_0>0$, we expect that
\[
\sigma_{\tilde{M}}(\tilde{P})\approx \sigma_{K}(\tilde{P}):=4\pi i\cdot\Tr[A \cap K, B \cap K, C \cap K]_{\widetilde{P}},
\]
where $K=A_{r}\cap B_{r}\cap C_{r}$ with $r\gg r_0$.

Assume that the sample $M$ is very large, in the sense that $M\supseteq K_{s}$ for some $s\gg  0$. We may think of $K$ as the ``bulk'' of the sample, far from the boundary $\partial M$. 
On $M$, the true Hamiltonian $H$ is formally the same operator as $\tilde{H}$, but with some local boundary conditions imposed on $\partial M$, and its Fermi projection is now $P\in\sB(M)$ (we may need to shift the cut-off energy $E$ slightly into a gap of the discrete spectrum of $H$). 
We have the 2-partition $(A \cap M, B \cap M, C \cap M)$ of $M$, but because $M$ is bounded, we would obtain
\[
\sigma_M(P)=4\pi i\cdot\Tr[A \cap M, B \cap M, C \cap M]_P=4\pi i\cdot \Tr[P_{A \cap M},P_{B \cap M}]=0.
\]
Actually, we are interested in the \emph{bulk} Hall conductance of $P$,
\[
\sigma_K(P):=4\pi i\cdot \Tr[A \cap K, B \cap K, C \cap K]_P.
\]
Note that $(A \cap K, B \cap K, C \cap K)$ is not a partition of $M$, so $\sigma_K(P)$ need not vanish. 

Indeed, a brief calculation gives
\begin{align*}
\sigma_K(P)&=12\pi i\cdot\Tr\big((KPK)A(KPK)B(KPK)-(KPK)B(KPK)A(KPK)\big)\\
&=12\pi i\cdot\Tr\big(K[PKAKP,PKBKP]K\big),
\end{align*}
which is \emph{not} the trace of a commutator. The underlying reason is that $KPK$ is no longer idempotent. Similarly for $\sigma_K(\tilde{P})$. So the quantization result does not apply to $\sigma_K(P)$ or $\sigma_K(\tilde{P})$.

Because $K$ is far from the boundary of $M$, under suitable assumptions on $H$ and $\tilde{H}$, we expect that $KPK\approx K\tilde{P}K$, and so
\[
\sigma_K(P)\approx\sigma_K(\tilde{P})\approx\sigma_{\tilde{M}}(\tilde{P})\in\ZZ.
\]
In other words, for a large finite-sized sample, the bulk Hall conductance should be \emph{approximately} integral. Indeed, as $K$ and $M$ are increased in size, we expect the approximation to get better.

The technical details of the above scheme has quite a different analytic flavour to this paper, so we postpone them to a future work. The reader is referred to \cite{Mitchell} for some calculations of approximately quantized $\sigma_K(P)$ in large finite-sized amorphous lattice systems.

\subsection{Plateaux}
The QHE expert will be aware that there are plateaux in the integral Hall conductivity even as the Fermi level $E$ is varied in certain intervals $I$. This is an important experimental signature, on top of the quantization. A general explanation is that the part of the spectrum of the Hamiltonian in the interval $I$ comprises only \emph{localized} eigenstates which do not contribute to the Hall conductance, so their inclusion/exclusion from the Fermi projection is immaterial. If the spectrum in $I$ is discrete, then the change of the Fermi projection $P=P_E$ as $E$ is varied in $I$ does not affect the pairing, as discussed in Section \ref{sec:trivial.pairing.localized}. 

However, if the whole interval $I$ has spectrum, there is no spectral gap available in $I$ for the construction of the Fermi projection by smooth functional calculus, only a so-called mobility gap. This could arise due to random disorder potentials inducing Anderson localization (dense pure point spectrum in $I$); see \cite{KLT-coarse} for a different mechanism involving ``geometrical dirt''. Nevertheless, for $M=\ZZ^2$ or $\RR^2$, the mobility gap condition is known to be closely related to a certain exponential decay of the integral kernel of $P_E$ (on average), e.g., \cite{AG}, albeit non-uniformly, in the ``weakly-local'' sense of \S3 in \cite{ST}. 
This would suggest that for general $M$, variants of our $\sB(M)$ can be used for the $P_E$ when $E$ varies in such a mobility gap. The technicalities involved are of a somewhat different nature to the main ideas of this paper, so we also leave this investigation for future work.

\section*{Acknowledgements}
G.C.T.\ thanks Y.\ Avron for pointing out \cite{Kitaev} during a discussion on coarse geometry.
G.C.T.\ and M.L.\ would like to thank Ulrich Bunke and Alexander Engel for helpful discussions.
M.L.\ acknowledges support from SFB 1085 “Higher invariants” funded by the German Research
Foundation (DFG).

\end{document}